\documentclass[12pt,a4paper,peerreview, onecolumn]{IEEEtran}
\usepackage{amsmath,amssymb,epsfig,verbatim,amsopn,subfigure,color,amsfonts}
\usepackage{xspace}
\usepackage{array,algorithm,algorithmic}
\usepackage{array}
\usepackage{multirow,array}%
\usepackage{multicol}%
\usepackage{makecell}
\usepackage[marginal]{footmisc}
\usepackage{slashbox,pict2e}
\usepackage{setspace}

\ifCLASSOPTIONpeerreview

\fi

\ifCLASSOPTIONonecolumn

\else
    
\fi

\newtheorem{remark}{Remark}
\newtheorem{theorem}{Theorem}
\newtheorem{proposition}{Proposition}

\newcommand{\dhat}[1]{\hat{\hat{#1}}}

\makeatletter
\renewcommand*{\@opargbegintheorem}[3]{\trivlist
  \item[\hskip \labelsep{\itshape #1\ #2}] {\itshape (#3):} {\normalfont}}
\makeatother
\newcommand{\AuthorOne}{Shama N. Islam} 
\newcommand{\AuthorTwo}{Salman~Durrani}
\newcommand{\AuthorThree}{Parastoo~Sadeghi}

\newcommand{\ThankOne}{The authors are with the Research School of Engineering, College of Engineering and Computer Science, The
Australian National University, Canberra, ACT 0200, Australia.
Emails: \{shama.islam, salman.durrani, parastoo.sadeghi\}@anu.edu.au.}

\begin{document}
\title{A Novel User Pairing Scheme for Functional Decode-and-Forward Multi-way Relay Network}
\author{\authorblockN{\AuthorOne, \AuthorTwo \,and  \AuthorThree \thanks{\ThankOne}}}

\maketitle

\begin{abstract}
In this paper, we consider a functional decode and forward (FDF)
multi-way relay network (MWRN) where a common user facilitates each
user in the network to obtain messages from all other users. We
propose a novel user pairing scheme, which is based on the principle
of selecting a common user with the best average channel gain. This
allows the user with the best channel conditions to contribute to
the overall system performance. Assuming lattice code based
transmissions, we derive upper bounds on the average common rate
and the average sum rate with the proposed pairing scheme.
Considering $M$-ary quadrature amplitude modulation with square constellation as a special
case of lattice code transmission, we derive asymptotic average symbol
error rate (SER) of the MWRN. We show that in terms of the
achievable rates, the proposed pairing scheme outperforms the
existing pairing schemes under a wide range of channel scenarios.
The proposed pairing scheme also has lower average SER compared to
existing schemes. We show that overall, the MWRN performance with
the proposed pairing scheme is more robust, compared to existing
pairing schemes, especially under worst case channel conditions when
majority of users have poor average channel gains.
\end{abstract}
\vspace{-10pt}

\section{Introduction}
Multi-way relay networks (MWRNs), where a single relay facilitates
all users in the network to exchange information with every other
user, have important potential applications in teleconferencing,
data exchange in a sensor network or file sharing in a social
network X
\cite{Ong:2010,Ong:2012,Ong:2014,Gunduz:2013,Gyan:2012,Gayan:2013,GWang:2012,Ma:2013,Noori:2012,Ang:2013,Tao:2013,Reza:2014}.
A MWRN is a generalization of two-way relay networks (TWRNs), which
enable bidirectional information exchange between two users and are
widely recognized in the literature for their improved spectral
efficiency, compared to conventional
relaying~\cite{Zhang-2006,Katti:2007,Laneman-2004,Popovski:2007,Cui:2009,Liew:2013}.
Note that multi-user
TWRNs~\cite{MChen:2009,MChen:2010,Xu:2011,Jianshu:2012,Hein:2013,NYang:2012},
where each user exchanges information with a pre-assigned user only,
can be considered as a special case of MWRNs.

The users in a MWRN can adopt either pairwise
transmission~\cite{Ong:2010,Gyan:2012,Noori:2012} or non-pairwise
transmission~\cite{Nazer:2011,Ma:2013,Gayan:2013,Gunduz:2013}
strategy for message exchange. Though non-pairwise transmission can
offer larger spectral efficiency, its benefits come at the expense
of additional signal processing complexity at the relay
\cite{Gayan:2013}. Hence, in this paper, we focus on pairwise
transmission strategy. Recently, pairwise transmission based MWRNs
have been studied for different relaying protocols, e.g., functional
decode and forward (FDF) \cite{Ong:2010}, decode and
forward~\cite{Gunduz:2013}, amplify and forward~\cite{Gyan:2012} and
compute and forward \cite{GWang:2012} protocols. It was shown
in~\cite{Ong:2010} that pairwise FDF with binary linear codes for
MWRN, where the relay decodes a function of the users' messages
rather than the individual messages from a user pair, is
theoretically the optimal strategy since it achieves the common rate. Also it was shown in~\cite{Ong:2012} that for a MWRN with
lattice codes in an Additive White Gaussian Noise (AWGN) channel,
the pairwise FDF achieves the common rate. Hence, in this
paper, we consider FDF MWRN.

In a pairwise transmission based FDF MWRN, user pair formation is a
critical issue. In this regard, two different pairing schemes have
been proposed in the literature. In the pairing scheme
in~\cite{Ong:2010}, the $\ell^{th}$ and the $(\ell+1)^{th}$ users
form a pair at the $\ell^{th}$ time slot, where $\ell\in[1,L-1]$
and $L$ is the number of users in the MWRN. In the pairing scheme in
\cite{Noori:2012}, instead of consecutive users as in the pairing
scheme in~\cite{Ong:2010}, the $\ell^{th}$ and the $(L-\ell+1)^{th}$
user form a pair at the $\ell^{th}$ time slot when
$1\leq\ell\leq\lfloor L/2\rfloor$ and the $(\ell+1)^{th}$ and
$(L-\ell+1)^{th}$ user form a pair at the $\ell^{th}$ time slot when
$\lfloor L/2\rfloor<\ell\leq L-1$, where $\lfloor\cdot\rfloor$
denotes the floor operation. The achievable rates for these two
existing pairing schemes were analyzed
in~\cite{Ong:2010,Ong:2012,Noori:2012}, while the average bit error
rate (BER) for the first pairing scheme was analyzed in
\cite{Shama:2012}. A major drawback of the above two pairing schemes
is that they arbitrarily select users for pair formation and do not take the users' channel information into account when pairing the users. This is crucial since in a MWRN, the
decision about each user depends on the decisions about all other
users transmitting before it. Thus, in the above pairing schemes, if
any user experiences poor channel conditions, it can lead to
incorrect detection of another user's message, which can adversely
impact the system performance due to error propagation. We also note that a recent paper on opportunistic pairing \cite{Tao:2013} also suffers from the error propagation problem similar to \cite{Ong:2010}.

In this paper, we propose a novel pairing scheme for user pair
formation in a FDF MWRN. In this scheme, each user is paired with a
common user, which is chosen by the relay as the user with the best
average channel gain. This allows the user with the best channel
conditions to contribute to improving the overall system performance
by reducing the error propagation in the network. The major
contributions of this paper are as follows:
\begin{itemize}

\item
Considering an $L$-user FDF MWRN employing sufficiently large dimension lattice codes, we derive upper bounds for the common rate and sum rate with the proposed pairing scheme (cf. Theorems $1-2$).

\item
Considering an $L$-user FDF MWRN with $M$-ary quadrature amplitude modulation (QAM) based transmission, which is a special case of lattice code based transmission, we derive the asymptotic average SER with the proposed pairing scheme (cf. Theorem $3$).

\item We present important insights, obtained from a careful analysis of the results in Theorems 1-3, in the form of Propositions 1-9. Analyzing the results in Theorems 1-3, we compare the performance of the proposed pairing scheme with the existing pairing schemes and show that:
\begin{itemize}
\item
For the equal average channel gain scenario, the average common rate and the average sum rate are the same for the proposed and existing pairing schemes, but the average SER improves with the proposed pairing scheme (cf. Propositions $1$, $4$ and $7$).

\item
For the unequal average channel gain scenario, the average common rate, the average sum rate and the average SER all improve for the proposed pairing scheme (cf. Propositions $2$, $5$ and $8$).

\item For the variable average channel gain scenario, the average common rate for the proposed pairing scheme is practically the same as the existing schemes, whereas, the average sum rate and the average SER improve for the proposed pairing scheme (cf. Propositions $3$, $6$ and $9$).
\end{itemize}

\end{itemize}

The rest of the paper is organized as follows. The system model is presented in Section~\ref{sec:sys_model}. The proposed pairing scheme is discussed in Section \ref{sec:pairing} and the general lattice code based transmissions with the proposed pairing scheme are presented in Section~\ref{sec:signal}. The common rate and the sum rate for a FDF MWRN with the proposed scheme is derived in Section \ref{performance}. The average SER is derived in Section \ref{error}. The numerical and simulation results for verification of the analytical solutions are provided in Section \ref{sec:results}. Finally, conclusions are provided in Section \ref{conclusion}.

Throughout this paper, we use the following notations: $\hat{(\cdot)}$ denotes the
estimate of a message, $\dhat{(\cdot)}$ denotes that the message is
estimated for the second time, $\mid\cdot\mid$ denotes absolute
value of a complex variable, $\|\cdot\|$ denotes Euclidean norm, $\arg{(\cdot)}$ denotes the argument,
$\max{(\cdot)}$ denotes the maximum value, $\min{(\cdot)}$ denotes
the minimum value, $E[\cdot]$ denotes the expected value of a random
variable, $\lfloor\cdot\rfloor$ denotes the floor operation,
$\log(\cdot)$ denotes logarithm to the base two and $Q(\cdot)$ is
the Gaussian Q-function.

\section{System Model}\label{sec:sys_model}
\begin{figure}
\centering
  \includegraphics[width=0.95\textwidth]{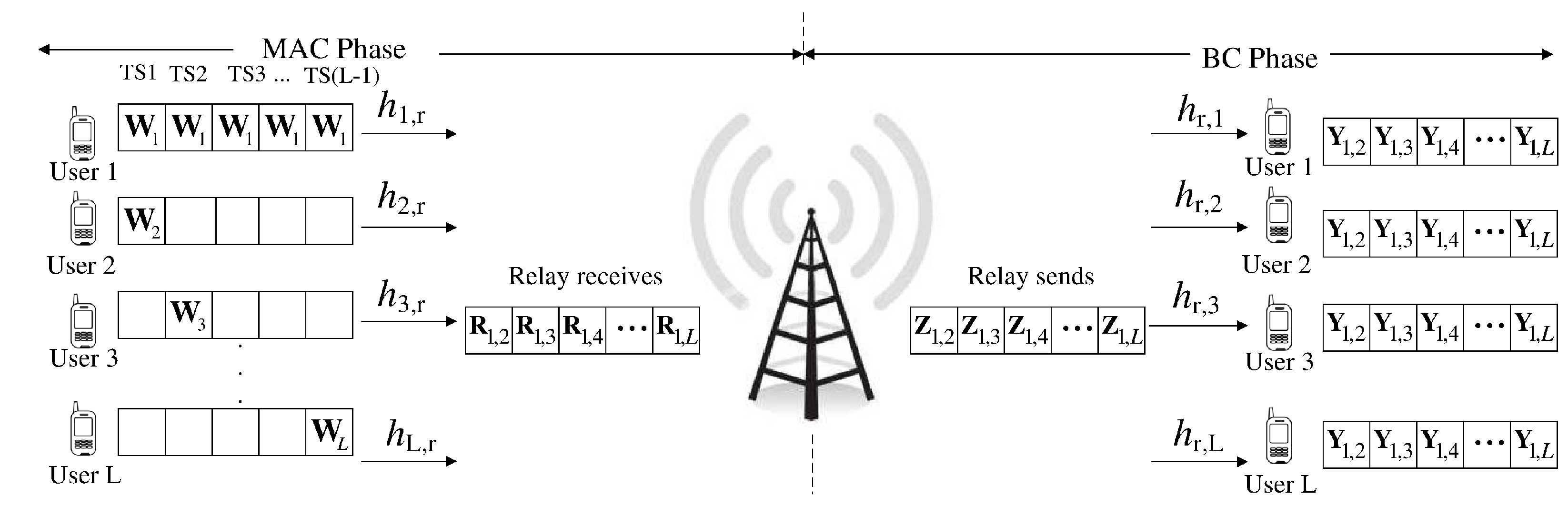}
  \caption{System model for an $L$-user multi-way relay network
(MWRN), where the users exchange information with each other via the
relay $R$. Here, `TS' means time slot and user $1$ is considered to
be the common user (for illustration purpose).} \label{fig:Fig1}
\end{figure}

We consider an $L$-user MWRN, where all the users exchange their
information with each other through a single relay, as illustrated
in Fig. \ref{fig:Fig1}. In this setup, a pair of users communicate
with each other at a time, while, the remaining users are silent. We
assume that the users transmit in a half-duplex manner and they do
not have any direct link in between them. The information exchange
takes place in two phases$-$multiple access and broadcast
phase$-$each comprising $L-1$ time slots for an $L$-user MWRN
\cite{Ong:2010}. In the \textit{multiple access phase}, the users
transmit their data in a pairwise manner. In the \textit{broadcast
phase}, the relay broadcasts the decoded network coded message to
all users. After $2(L-1)$ time slots, all users have the network
coded messages corresponding to each user pair and then they utilize
self information to extract the messages of all the other users. We
refer to these $2(L-1)$ time slots in the two phases as one
\textit{time frame}. That is, in each time frame, each user
transmits a message packet of length $T$ and the relay transmits
$(L-1)$ message packets, each of length $T$. Thus, a total of
$(2L-1)$ message packets are communicated in an entire time frame.
We choose the index for time slot and time frame as $t_s$ and $t_f$,
respectively, and the message index as $t$ where, $t_s\in[1,L-1]$,
$t\in[1,T]$ and $t_f\in[1,F]$, where, $F$ is the total number of
time frames. The transmission power of each user is $P$, whereas,
the transmission power of the relay is $P_r$. At the $t_f^{th}$ time
frame and the $t_s^{th}$ time slot, the channel from the $j^{th}$
user to the relay is denoted by $h_{j,r}^{t_s,t_f}$ and the channel
from the relay to the $j^{th}$
user by $h_{r,j}^{t_s,t_f}$, where $j\in[1,L]$. 
We make the following assumptions regarding the
channels:

\begin{itemize}
\item The channels are assumed to be block Rayleigh fading channels, which remain
constant during one message packet transmission in a certain time
slot in a certain multiple access or broadcast phase. The channels
in different time slots (e.g., $h_{1,r}^{1,1}$ and $h_{1,r}^{2,1}$) and
different time frames (e.g., $h_{1,r}^{1,1}$ and $h_{1,r}^{1,2}$) are considered to be independent. Also, the channels from users to the relay (e.g., $h_{j,r}^{t_s,t_f}$) and the channels
from the relay to users (e.g., $h_{r,j}^{t_s,t_f}$) are reciprocal.
\item The fading channel coefficients are zero mean
complex-valued Gaussian random variables with variances
$\sigma^{2}_{h_{j,r}}=\sigma^{2}_{h_{r,j}}$.

\item The perfect instantaneous channel state information (CSI) of all users is available to the relay. The users have access to the self CSI only, which has been assumed in many research works~\cite{MJu:2010,Louie-2010,Zhao:2011}.
\item Perfect channel phase synchronization is assumed because physical layer network coding requires that the signals arrive at the relay with the same phase and this allows benchmark performance to be determined \cite{Zhang-2006,Cui:2009}.
\end{itemize}

We consider the following three different channel scenarios in this work:

\begin{enumerate}
\item \emph{Equal average channel gain scenario:} All the channels from the relay to the users and the users
to the relay have equal average channel gain, which remain fixed for
all time frames. That is, $E[\mid h_{1,r}^{t_s,t_f}\mid^2]=E[\mid
h_{2,r}^{t_s,t_f}\mid^2]=...=E[\mid h_{L,r}^{t_s,t_f}\mid^2]$.

\item \emph{Unequal average channel gain scenario:} All the channels from the relay to the users and the users
to the relay have unequal average channel gains which remain fixed
for all the time frames. That is, $E[\mid
h_{1,r}^{t_s,t_f}\mid^2]\neq E[\mid
h_{2,r}^{t_s,t_f}\mid^2]\neq...\neq E[\mid h_{L,r}^{t_s,t_f}\mid^2]$
and $E[\mid h_{j,r}^{t_s,1}\mid^2]=E[\mid
h_{j,r}^{t_s,2}\mid^2]=...=E[\mid h_{j,r}^{t_s,F}\mid^2]$. 

\item \emph{Variable average channel gain scenario:} All the channels from the relay to the users and the users
to the relay have unequal average channel gains and the channel
conditions change after a block of $T'_f$
($T'_f<F$) time frames. That is, $E[\mid
h_{1,r}^{t_s,t_f}\mid^2]\neq E[\mid
h_{2,r}^{t_s,t_f}\mid^2]\neq...\neq E[\mid h_{L,r}^{t_s,t_f}\mid^2]$
and $E[\mid h_{j,r}^{t_s,aT'_f+1}\mid^2]=E[\mid
h_{j,r}^{t_s,aT'_f+2}\mid^2]=...=E[\mid
h_{j,r}^{t_s,(a+1)T'_f}\mid^2]$ for $j\in[1,L]$ and $0\leq a\leq
\frac{F}{T'_f}-1$, where $T'_f$ is the number of time frames after
which the unequal average channel gains change.
\end{enumerate}

The above scenarios can model a wide variety of practical channel
scenarios. For example, the equal average channel gain scenario is
applicable to satellite communications, where the users are
equidistant from the relay. The unequal average channel gain
scenario is applicable to fixed users (e.g., located at home or
workplace) in a network, where the users' distances from the relay
are unequal but remain fixed. The variable average channel gain
scenario is applicable to mobile users in a network, where the
users' distances from the relay are unequal and vary due to user
mobility.

\section{Proposed Pairing Scheme for MWRN} \label{sec:pairing}
In this section, we propose a new pairing scheme for user pair
formation in the multiple access phase (illustrated in Fig.
\ref{fig:Fig1}) which is defined by the following set of
principles:

\begin{description}
\item[P1] The common user is selected by the relay to be the user that has the best average channel gain in the system.
\item[P2] The common user's index is broadcast by the relay prior to each multiple access phase.
This common user transmits in all the time slots in the multiple
access phase and the other users take turns to form a pair with
this common user.
\item[P3] The common user is kept fixed for all the time slots within a certain time frame. After some time frames, the common user might change depending upon the changing channel conditions. 
\end{description}

The proposed pairing scheme allows the best channel in the system to
contribute towards the error-free detection of each user's message,
which would not be possible if the common user is chosen without
considering the channel conditions, as in
\cite{Ong:2010,Noori:2012}. Note that taking channel state information into account is a well established design principle in wireless communication systems \cite{Alouini:2000}.

In the propsed scheme, since the common user is involved in all the transmissions in the multiple access phase, an issue of transmission fairness arises. In the context of the proposed scheme, on average, each user should transmit the same number of times (equivalently consume the same amount of power overall). We propose to achieve transmission fairness for the three channel scenarios, considered in this work, in the following manner

\begin{enumerate}
\item \emph{Equal average channel gain scenario:} In this scenario, 
to maintain transmission fairness among the users, we randomly select a different common user in each time frame so that, on average, every user gets the opportunity to become the common user.

\item \emph{Unequal average channel gain scenario:} In this scenario, 
the common user's transmission power must be scaled by $(L-1)$, since it transmits $(L-1)$ times, whereas, other users transmit only once.

\item \emph{Variable average channel gain scenario:} In this scenario, 
during each time frame, the user with the best average channel gain is chosen as the common user and this process is repeated for every time frame so that, on average, every user with changing channel conditions, gets the opportunity to become the common user.

\end{enumerate}

\section{Signal Transmissions With the Proposed Pairing Scheme}\label{sec:signal}

In this section, we discuss the general lattice code based transmissions with the proposed pairing scheme in a MWRN. We denote the $i^{th}$ user as the common user and the $\ell^{th}$ user as the other users, where, $i,\ell\in[1,L]$ and $\ell\neq i$. For the rest of this paper, we consider message exchange within a certain time frame and choose to omit the superscript $t_f$ from the symbols for simplifying the notations.

\subsection{Preliminaries on Lattice Codes}
As our proposed pairing scheme is based on lattice codes, we first
present the definitions of some primary operations on lattice codes,
which we have used in the later subsections. Our notations for
lattice codes follow those of \cite{Ma:2013,Ong:2012}. Further
details on lattice codes are available in
\cite{Erez:2004,Nazer:2011,Nam:2011,YSong:2013}.

An $N$-dimensional lattice is a discrete subgroup of the
$N$-dimensional complex field $\mathbb{C}^N$ under the normal vector
addition operation and can be expressed as
\cite{Ma:2013,YSong:2013}:
\begin{equation}\label{eq:latiice_def}
\Lambda=\{\lambda=\mathbf{G}_\Lambda c: c\in \mathbb{Z}^N\}
\end{equation}

\noindent where, $\mathbf{G}_\Lambda\in\mathbb{C}^{N\times N}$ is
the generator matrix corresponding to the lattice $\Lambda$ and
$\mathbb{Z}$ is the set of integers.

\begin{itemize}
\item The nearest neighbour lattice quantizer maps a point
$\mathbf{x}\in\mathbb{C}^N$ to a nearest lattice point
$\mathbf{\lambda}\in\Lambda$ in Euclidean distance \cite{Ma:2013}. That is,
\begin{equation}\label{eq:quantizer}
Q_\Lambda(\mathbf{x})=\arg\min_\mathbf{\lambda}\| \mathbf{x}-\mathbf{\lambda}\|^2
\end{equation}
\item The modulo-$\Lambda$ operation is defined by
$\mathbf{x}\mod\Lambda=\textbf{x}-Q_\Lambda(\textbf{x})$
\cite{Nam:2011,Erez:2004,Ong:2012,YSong:2013}.
\item The Voronoi region $\mathcal{V}(\Lambda)$ denotes
the set of all points in the $N$-dimensional complex field
$\mathbb{C}^N$, which are closest to the zero vector \cite{Ma:2013}, i.e.,
\begin{equation}\label{eq:voronoi}
\mathcal
{V}(\Lambda)=\{\mathbf{x}\in\mathbb{C}^N:Q_\Lambda(\mathbf{x})=\mathbf{0}\},
\end{equation}
\item $\psi(\cdot)$ denotes the mapping of messages from a finite dimensional field to lattice points,
i.e., $\psi(\mathbf{w})\in \Lambda$, where $\mathbf{w}$ is a message from a finite dimensional field.
\item A coarse lattice $\Lambda$ is nested in a fine lattice $\Lambda_f$, i.e., $\Lambda\subseteq\Lambda_f$, so that the messages mapped into fine lattice points remain in the voronoi region of the coarse lattice.
\item The dither vectors $\mathbf{d}$ are generated
independently from a uniform distribution over the fundamental
Voronoi region $\mathcal V(\Lambda)$.
\end{itemize}
\subsection{Multiple Access Phase}

In this phase, the common user and one other user transmit simultaneously using FDF based on lattice codes and the relay receives the sum of the signals, i.e., at the $(\ell-1)^{th}$ time slot, users $i$ and $\ell$ transmit simultaneously.

\subsubsection{Communication Protocol at the Users}

In a certain time frame, the message packet of the $\ell^{th}$ user
is denoted by
\begin{equation}
\textbf{W}_{\ell}^{t_s}=\left\{\begin{array}{ll}\{W_\ell^{t_s,1},W_\ell^{t_s,2},...,
W_\ell^{t_s,T}\} &\mbox {$t_s=\ell-1$}\\ 0 &\mbox
{$t_s\neq\ell-1$},\end{array}\right.\end{equation}

\noindent where, the elements $W_\ell^{t_s,t}$ are generated
independently and uniformly over a finite field. Similarly, the
message packet of the $i^{th}$ user is given by
$\textbf{W}_{i}=\{W_i^{t_s,1},W_i^{t_s,2},..., W_i^{t_s,T}\}$ for
$t_s\in[1,L-1]$.

During a certain time frame, in the $t_s=(\ell-1)^{th}$ time slot,
the $i^{th}$ user and the $\ell^{th}$ user transmit their messages
using lattice codes
$\textbf{X}_i=\{X_i^{t_s,1},X_i^{t_s,2},...,X_i^{t_s,T}\}$ and
$\textbf{X}_\ell^{t_s}=\{X_{\ell}^{t_s,1},X_{\ell}^{t_s,2},...,X_{\ell}^{t_s,T}\}$,
respectively, which can be given by \cite{Ong:2012,Gunduz:2013}:
\begin{subequations}\label{eq:lattice_code}
\begin{align}
X_{i}^{t_s,t}&=(\psi(W_i^{t_s,t})+d_i)\mod
\Lambda, \\
X_{\ell}^{t_s,t}&=(\psi(W_\ell^{t_s,t})+d_\ell)\mod\Lambda,
\end{align}
\end{subequations}

\noindent where, $d_i$ and $d_\ell$ are the
dither vectors for the $i^{th}$ and the $\ell^{th}$ user. The dither vectors are generated at the users and transmitted to the relay
prior to message transmission in the multiple access phase
\cite{Ma:2013}. 

\subsubsection{Communication Protocol at the Relay} The relay
receives the signal
$\textbf{R}_{i,\ell}^{t_s}=\{r_{i,\ell}^{t_s,1},r_{i,\ell}^{t_s,2},...,r_{i,\ell}^{t_s,T}\}$,
where
\begin{equation}\label{eq:received}
r_{i,\ell}^{t_s,t}=\sqrt{P}h_{i,r}^{t_s}X_{i}^{t_s,t}+\sqrt{P}h_{\ell,r}^{t_s}X_{\ell}^{t_s,t}+n_1,
\end{equation}

\noindent where $n_{1}$ is the zero mean complex AWGN at the relay with noise variance $\sigma^2_{n_1}=\frac{N_0}{2}$ per dimension and $N_0$ is the noise power. 

\subsection{Broadcast Phase}

In this phase, the relay broadcasts the decoded network
coded message and each user receives it.

\ifCLASSOPTIONpeerreview
\vspace{-0.1cm}
\fi

\subsubsection{Communication Protocol at the Relay}
The relay scales the received signal with a scalar coefficient
$\alpha$ \cite{Nazer:2011} and removes the dithers $d_i,d_\ell$
scaled by $\sqrt{P}h_{i,r}^{t_s}$ and $\sqrt{P}h_{\ell,r}^{t_s}$,
respectively. The resulting signal is given by
{\fontsize{11}{13.2}\selectfont
\begin{align}\label{eq:scaling}
X_r^{t_s,t}&=[\alpha
r_{i,\ell}^{t_s,t}-\sqrt{P}h_{i,r}^{t_s}d_i-\sqrt{P}h_{\ell,r}^{t_s}
d_\ell]\mod\Lambda\nonumber\\
&=[\sqrt{P}h_{i,r}^{t_s}X_i^{t_s,t}+\sqrt{P}h_{\ell,r}^{t_s}
X_\ell^{t_s,t}+(\alpha-1)\sqrt{P}(h_{i,r}^{t_s}X_i^{t_s,t}+
h_{\ell,r}^{t_s}X_\ell^{t_s,t})+\alpha
n_1-\sqrt{P}h_{i,r}^{t_s}d_i-\sqrt{P}h_{\ell,r}^{t_s}
d_\ell]\nonumber\mod\Lambda\nonumber\\
&=[\sqrt{P}h_{i,r}^{t_s}\psi(W_i^t)+\sqrt{P}h_{\ell,r}^{t_s}
\psi(W_\ell^{t_s,t})+n]\mod\Lambda,
\end{align}\par
}
\noindent where, $n=(\alpha-1)\sqrt{P}(h_{i,r}^{t_s}X_i^t+
h_{\ell,r}^{t_s}X_\ell^{ts,t})+\alpha n_1$ and $\alpha$ is chosen to
minimize the noise variance \cite{Nam:2011,Erez:2004}. 

The relay decodes the signal in \eqref{eq:scaling} with a lattice
quantizer \cite{Erez:2004,Nazer:2011} to obtain an estimate
$\hat{\textbf{V}}_{i,\ell}^{t_s}=\{\hat{V}_{i,\ell}^{t_s,1},\hat{V}_{i,\ell}^{t_s,2},...,\hat{V}_{i,\ell}^{t_s,T}\}$
which is a function of the messages $\textbf{W}_i$ and
$\textbf{W}_\ell^{t_s}$. Since, for sufficiently large $N$,
Pr$(n\notin\mathcal{V})$ approaches zero \cite{Nazer:2011},
$\hat{\textbf{V}}_{i,\ell}^{t_s}=(\psi(\textbf{W}_{i})+\psi(\textbf{W}_{\ell}^{t_s}))\mod
\Lambda$. The relay then adds a dither $d_r$ with the network coded
message which is generated at the relay and broadcast to the users
prior to message transmission in the broadcast phase \cite{Ma:2013}.
Then it broadcasts the resulting message using lattice codes, which
is given as
$\textbf{Z}_{i,\ell}^{t_s}=\{Z_{i,\ell}^{t_s,1},Z_{i,\ell}^{t_s,2},...,Z_{i,\ell}^{t_s,T}\}$,
where
$Z_{i,\ell}^{t_s,t}=(\hat{{V}}_{i,\ell}^{t_s,t}+d_r)\mod\Lambda$.


\subsubsection{Communication Protocol at the Users}
The $j^{th}$ user receives
$\textbf{Y}_{i,\ell}^{t_s}=\{Y_{i,\ell}^{t_s,1},Y_{i,\ell}^{t_s,2},...,Y_{i,\ell}^{t_s,T}\}$,
where
\begin{equation}\label{eq:dfdetect}
    Y_{i,\ell}^{t_s,t}=\sqrt{P_r}h_{r,j}^{t_s}Z_{i,\ell}^{t_s,t}+n_{2},
\end{equation}

\noindent and $n_{2}$ is the zero mean complex AWGN at the user with
noise variance $\sigma^2_{n_2}=\frac{N_0}{2}$ per dimension. 
At the end of the broadcast phase, the $j^{th}$ user scales the
received signal with a scalar coefficient $\beta_j$ and removes the
dithers $d_r$ multiplied by $\sqrt{P_r}h_{r,j}$. The resulting signal
is
%
{\fontsize{11}{13.2}\selectfont
\begin{align}\label{eq:rec_scaled}
[\beta_jY_{i,\ell}^{t_s,t}-\sqrt{P_r}h_{r,j}^{t_s}d_r]\mod\Lambda =
&[\sqrt{P_r}h_{r,j}^{t_s}\hat{V}^{t_s,t}_{i,\ell}+(\beta_j-1)\sqrt{P_r}h_{r,j}^{t_s}\hat{V}^{t_s,t}_{i,\ell}+\beta_jn_2]\mod\Lambda  \nonumber \\
&=[\sqrt{P_r}h_{r,j}^{t_s}\hat{V}_{i,\ell}^{t_s,t}+n']\mod\Lambda,
\end{align}\par}

\noindent where,
$n'=\sqrt{P_r}h_{r,j}^{t_s}(\beta_j-1)\hat{V}^{t_s,t}_{i,\ell}+\beta_jn_2$
and $\beta_j$ is chosen to minimize the noise variance
\cite{Ma:2013}. The users then detect the received signal with a
lattice quantizer \cite{Ma:2013} and obtain the estimate
$\dhat{\textbf{V}}_{i,\ell}^{t_s}=(\psi(\textbf{W}_{i})$
$+\psi(\textbf{W}_{\ell}^{t_s}))\mod\Lambda$, assuming that the
lattice dimension is large enough such that
Pr$(n'\notin\mathcal{V})$ approaches zero. After decoding all the
network coded messages, each user performs message extraction of
every other user by canceling self information.
\ifCLASSOPTIONpeerreview
\vspace{-0.1cm}
\fi
\subsubsection{Message Extraction at the Common User}
For the common user ($i^{th}$ user), this message extraction
involves simply subtracting the lattice point corresponding to its
own message from the lattice network coded messages
$\dhat{\textbf{V}}_{i,\ell}^{t_s}$. The process can be shown as %
%
\begin{align}
\psi(\hat{\textbf{W}}_{\ell}^{t_s})&=(\dhat{\textbf{V}}_{i,\ell}^{t_s}-
\psi(\textbf{W}_{i})) \mod\Lambda, \quad
\ell\in[1,L], \ell\neq i.
\end{align}

\ifCLASSOPTIONpeerreview
\vspace{-0.1cm}
\fi
\subsubsection{Message Extraction at the Other Users}
For other users, the process is different from the common user. At
first, the $\ell^{th}$ user subtracts the scaled lattice point
corresponding to its own message, i.e., $\psi(\textbf{W}_{\ell}^{t_s})$ from the network
coded message received in the $(\ell-1)^{th}$ time slot (i.e.,
$\dhat{\textbf{V}}_{i,\ell}^{t_s}$) and extracts the message of the
$i^{th}$ user as $\psi(\hat{\textbf{W}}_{i})$.
After that, it utilizes the extracted message of the $i^{th}$ user
to obtain the messages of other users in a similar manner. The
message extraction process in this case can be shown as
%
\begin{align}\label{eq:otherextract}
\psi(\hat{\textbf{W}}_{i})=(\dhat{\textbf{V}}_{i,\ell}^{t_s}-
\psi(\textbf{W}_{\ell}^{t_s})) \mod\Lambda,
\psi(\hat{\textbf{W}}_{m}^{t_s})=(\dhat{\textbf{V}}_{i,m}^{t_s}-
\psi(\hat{\textbf{W}}_{i}))\mod\Lambda,\quad m\in[1,L],m\neq i,\ell.
\end{align}
%
\section{Common Rate and Sum Rate Analysis}\label{performance}
In this section, we investigate common rate and sum rate of
the MWRN with the proposed pairing scheme. We first analyze the SNR
of each user pair in a MWRN and use these results to obtain
expressions for the achievable rates. For the rest of this paper, we
simplify the notations by omitting the time slot superscript $t_s$.

\subsection{SNR analysis}\label{SNR:df}
In a FDF MWRN, the decoding operation is performed after both the
multiple access phase and the broadcast phase. Thus, we need to
consider the SNR at the users and the SNR at the relay, separately.
%

\subsubsection{SNR at the Users}
The SNR at the users have the same expressions for all the three
pairing schemes. The signal transmission from the relay to any user
$j\in[1,L]$ is the same as that in a point-to-point fading channel.
Thus, the SNR of the $m^{th} (m\in [1,L])$ user's signal received at
the $j^{th}$ user is given by:
\begin{equation}\label{eq:snrlatreceived}
\gamma_{j}=\frac{P_r\mid h_{r,j}\mid^2}{\mid\beta_j\mid^2
N_0+P_r\mid\beta_j-1\mid^2\mid h_{r,j}\mid^2}.
\end{equation}
\noindent where, the numerator represents power of the signal part in \eqref{eq:rec_scaled} and the denominator represents the power of the noise term $n'$ in \eqref{eq:rec_scaled}.

\subsubsection{SNR at the Relay}

In a FDF MWRN based on lattice coding with the proposed pairing
scheme, the SNR of the received signal at the relay can be obtained from
\eqref{eq:scaling} as
\begin{equation}\label{eq:snrlatticenp}
\gamma_{r}(i,\ell)=\frac{P\min(\mid h_{i,r}\mid^2,\mid
h_{\ell,r}\mid^2)}{\mid\alpha\mid^2 N_0+P\mid\alpha-1\mid^2 (\mid
h_{i,r}\mid^2+\mid h_{\ell,r}\mid^2)},
\end{equation}

\noindent where, the numerator represents the power of the signal
part (i.e., $\sqrt{P}h_{i,r}^{t_s}\psi(W_i^t)+\sqrt{P}h_{\ell,r}^{t_s}
\psi(W_\ell^{t_s,t})$ in \eqref{eq:scaling}) and the denominator represents the power of
the noise terms $n$ in \eqref{eq:scaling}.

For the pairing scheme in \cite{Ong:2010}, the SNR received at the
relay can be expressed as
\begin{equation}\label{eq:snrlatticeong}
\gamma_{r}(i)=\frac{P\min(\mid h_{\ell,r}\mid^2,\mid
h_{\ell+1,r}\mid^2)}{\mid\alpha\mid^2 N_0+P\mid\alpha-1\mid^2 (\mid
h_{\ell,r}\mid^2+\mid h_{\ell+1,r}\mid^2)}.
\end{equation}

Similarly, for the pairing scheme in \cite{Noori:2012}, the SNR at
the relay is given by
\begin{equation}\label{eq:snrlatticenoori}
\gamma_{r}(i)=\frac{P\min(\mid h_{\ell,r}\mid^2,\mid
h_{L-\ell+2,r}\mid^2)}{\mid\alpha\mid^2 N_0+P\mid\alpha-1\mid^2
(\mid h_{\ell,r}\mid^2+\mid h_{L-\ell+2,r}\mid^2)}.
\end{equation}

Note that \eqref{eq:snrlatticenp}, \eqref{eq:snrlatticeong} and
\eqref{eq:snrlatticenoori} have the same form and differ in the
indices of the channel coefficients, which is determined by the
pairing scheme.

\subsection{Common Rate}\label{capacity}

Common rate indicates the maximum possible information rate
of the system that can be exchanged with negligible error. It can be
a useful metric for the systems where all the users have the same
amount of
information to exchange \cite{Ong:2012}. 

Assuming lattice codes with sufficiently large dimensions are
employed, the common rate for an $L$-user FDF MWRN is given
by \cite{Ong:2010,Noori:2012}
\begin{equation}\label{eq:Rc}
R_{c}=\frac{1}{L-1}\min_{\ell-1\in[1,L-1]}\{R_{c,\ell-1}\},
\end{equation}

\noindent where, the factor $\frac{1}{L-1}$ is due to the fact that
the complete message exchange requires $L-1$ time slots and
$R_{c,\ell-1}$ is the achievable rate in the $(\ell-1)^{th}$ time
slot, given by
\begin{equation}\label{eq:R}
R_{c,\ell-1}=\min\{R_{M,\ell-1},R_{B,\ell-1}\},
\end{equation}

\noindent where, $R_{M,\ell-1}$ and $R_{B,\ell-1}$ are the maximum
achievable rates at the $(\ell-1)^{th}$ time slot during the
multiple access phase and the broadcast phase, respectively. Next,
we derive the upper bounds on the maximum achievable rates in the
multiple access and broadcast phases.

\begin{theorem}\label{th:cap} For the proposed pairing scheme in a FDF MWRN, the maximum achievable rate during the $(\ell-1)^{th}$ time slot in
the multiple access phase is upper bounded by
\begin{equation}\label{eq:Rm}
R_{M,\ell-1}\leq \frac{1}{2}\log\left(\min\left(\frac{\mid
h_{i,r}\mid^2}{\mid h_{i,r}\mid^2+\mid h_{\ell,r}\mid^2}+\frac{P\mid
h_{i,r}\mid^2}{N_0},\frac{\mid h_{\ell,r}\mid^2}{\mid
h_{i,r}\mid^2+\mid h_{\ell,r}\mid^2}+\frac{P\mid
h_{\ell,r}\mid^2}{N_0}\right)\right),
\end{equation}

\noindent and the maximum achievable rate during the $(\ell-1)^{th}$
time slot in the broadcast phase is upper bounded by
\begin{equation}\label{eq:Rb}
R_{B,\ell-1}\leq \frac{1}{2}\log\left(1+\frac{\min_{j\in[1,L]}\mid
h_{j,r}\mid^2P_r}{N_0}\right).
\end{equation}
\end{theorem}

\begin{proof}
See Appendix A.
\end{proof}


Note that the common rate for the pairing scheme in \cite{Ong:2010}
and in \cite{Noori:2012} can be obtained by replacing the subscript
$i$ with $\ell-1$ and $L-\ell+2$, respectively in \eqref{eq:Rm} and
using \eqref{eq:Rb}, \eqref{eq:R} and \eqref{eq:Rc}.

Using Theorem \ref{th:cap} and substituting in \eqref{eq:R} and
\eqref{eq:Rc}, the average common rate for the proposed pairing
scheme can be given as in \eqref{eq:Rcnew},
\begin{figure*}
{\fontsize{11}{13.2}\selectfont
\begin{subequations}
\begin{align}
E[R_{c}]&\leq \frac{1}{2(L-1)}E\left[\log\left(\min\left(\frac{1}{1+
\frac{\mid h_{\ell,r}\mid^2}{\mid h_{i,r}\mid^2}}+\frac{P\mid
h_{i,r}\mid^2}{N_0},\frac{1}{1+\frac{\mid h_{i,r}\mid^2}{\mid
h_{\ell,r}\mid^2}}+\frac{P\mid h_{\ell,r}\mid^2}{N_0}\right)\right)\right]\label{eq:R1}\\
&\leq \frac{1}{2(L-1)}\log \left(E\left[\min\left(\frac{1}{1+
\frac{\mid h_{\ell,r}\mid^2}{\mid h_{i,r}\mid^2}}+\frac{P\mid
h_{i,r}\mid^2}{N_0},\frac{1}{1+\frac{\mid h_{i,r}\mid^2}{\mid
h_{\ell,r}\mid^2}}+\frac{P\mid h_{\ell,r}\mid^2}{N_0}\right)\right]\right)\label{eq:Jensen}\\
&\leq \frac{1}{2(L-1)}\log \left(\min\left( E\left[\frac{1}{1+
\frac{\mid h_{\ell,r}\mid^2}{\mid h_{i,r}\mid^2}}+\frac{P\mid
h_{i,r}\mid^2}{N_0}\right],E\left[\frac{1}{1+\frac{\mid
h_{i,r}\mid^2}{\mid
h_{\ell,r}\mid^2}}+\frac{P\mid h_{\ell,r}\mid^2}{N_0}\right]\right)\right)\label{eq:mean}\\
&= \frac{1}{2(L-1)}\log\left(\min\left(\frac{1}{1+ \frac{\sigma_{
h_{\ell,r}}^2}{\sigma_
{h_{i,r}}^2}}+\frac{P\sigma_{h_{i,r}}^2}{N_0},\frac{1}{1+\frac{\sigma_{
h_{i,r}}^2}{\sigma_{
h_{\ell,r}}^2}}+\frac{P\sigma_{h_{\ell,r}}^2}{N_0}\right)\right),\label{eq:Rcnew}
\end{align}
\end{subequations}\par}
\hrulefill
\end{figure*}
\noindent where the inequality in \eqref{eq:Jensen} holds from
Jensen's inequality and the inequality \eqref{eq:mean} comes from
the fact that $E[\min(A_1,A_2)]\leq E[A_1],E[A_2]$, where $A_1,A_2$ are independent
random variables.

Similarly, the average common rate for the pairing scheme in
\cite{Ong:2010} can be expressed as
{\fontsize{11}{13.2}\selectfont
\begin{equation}\label{eq:Rcong}
E[R_{c}]\leq \frac{1}{2(L-1)}\log\left(\min\left(\frac{1}{1+
\frac{\sigma_{ h_{\ell,r}}^2}{\sigma_
{h_{\ell-1,r}}^2}}+\frac{P\sigma_{h_{\ell-1,r}}^2}{N_0},\frac{1}{1+\frac{\sigma_{
h_{\ell-1,r}}^2}{\sigma_{
h_{\ell,r}}^2}}+\frac{P\sigma_{h_{\ell,r}}^2}{N_0}\right)\right),
\end{equation}
}

\noindent and the average common rate for the pairing scheme in
\cite{Noori:2012} can be given as
{\fontsize{11}{13.2}\selectfont
\begin{align}\label{eq:Rcnoori}
E[R_{c}]\leq \frac{1}{2(L-1)}&\log\left(\min\left(\frac{1}{1+
\frac{\sigma_{ h_{\ell-1,r}}^2}{\sigma_
{h_{L-\ell+2,r}}^2}}+\frac{P\sigma_{h_{L-\ell+2,r}}^2}{N_0},\frac{1}{1+\frac{\sigma_{
h_{L-\ell+2,r}}^2}{\sigma_{
h_{\ell-1,r}}^2}}+\frac{P\sigma_{h_{\ell-1,r}}^2}{N_0},\right.\right.\nonumber\\
&\left.\left. \frac{1}{1+ \frac{\sigma_{ h_{\ell,r}}^2}{\sigma_
{h_{L-\ell+2,r}}^2}}+\frac{P\sigma_{h_{L-\ell+2,r}}^2}{N_0},\frac{1}{1+\frac{\sigma_{
h_{L-\ell+2,r}}^2}{\sigma_{
h_{\ell,r}}^2}}+\frac{P\sigma_{h_{\ell,r}}^2}{N_0}\right)\right).
\end{align}\par
}

While \eqref{eq:Rcnew}$-$\eqref{eq:Rcnoori} do not provide tight upper
bounds on the average common rate, they allow an analytical
comparison of the proposed and existing pairing schemes. The main
results from the analytical comparison are summarized in the Propositions 1$-$3. Note that in
Section \ref{sec:results}, the actual expressions of the
instantaneous rates are averaged over a large number of channel
realizations to corroborate the insights presented in Propositions 1$-$3.

\begin{proposition}\label{prop:eq_cap}
The average common rate for the proposed pairing scheme and the pairing schemes in \cite{Ong:2010} and~\cite{Noori:2012} are the same for the equal average channel gain scenario.
\end{proposition}

\begin{proof}
See Appendix B.
\end{proof}

\begin{proposition}\label{prop:uneq_cap}
The average common rate for the proposed pairing scheme is
larger than that of the pairing schemes in
\cite{Ong:2010} and \cite{Noori:2012} for the unequal average channel gain
scenario.
\end{proposition}

\begin{proof}
See Appendix B.
\end{proof}

\begin{proposition} \label{prop:var_cap}
The average common rate for the proposed pairing scheme is
practically the same as that of the pairing schemes in
\cite{Ong:2010} and \cite{Noori:2012}, for the variable
average channel gain scenario.
\end{proposition}

\begin{proof}
See Appendix B.
\end{proof}

\subsection{Sum Rate}\label{sec:sum_rate}
The sum rate indicates the maximum throughput of the system. For a
FDF MWRN, the sum rate can be defined as the sum of the achievable
rates of all users for a complete round of information exchange.

\begin{theorem}\label{th:sum}
For the proposed pairing scheme in a FDF MWRN, the sum rate is given
by:
\begin{equation}\label{eq:Rs}
R_{s}= \frac{1}{2(L-1)}\sum_{\ell=1, \ell\neq
i}^{L}\left(\log\left(\frac{\mid h_{i,r}\mid^2}{\mid
h_{i,r}\mid^2+\mid h_{\ell,r}\mid^2}+\frac{P\mid
h_{i,r}\mid^2}{N_0}\right)+\log\left(\frac{\mid
h_{\ell,r}\mid^2}{\mid h_{i,r}\mid^2+\mid
h_{\ell,r}\mid^2}+\frac{P\mid h_{\ell,r}\mid^2}{N_0}\right)\right).
\end{equation}
\end{theorem}

\begin{proof}
see Appendix C.
\end{proof}

Note that the sum rate for the pairing scheme in \cite{Ong:2010} and the
pairing scheme in \cite{Noori:2012} can be obtained by replacing the
subscript $i$ with $\ell-1$ and $L-\ell+2$, respectively in
\eqref{eq:Rs}. 
%

Using Theorem \ref{th:sum}, the average sum rate (averaged over all
channel realizations) for the proposed pairing scheme can be given
as in \eqref{eq:Rsavg}, using similar steps as in \eqref{eq:R1}, \eqref{eq:Jensen} and \eqref{eq:Rcnew}.
\begin{figure*}
{\fontsize{11}{13.2}\selectfont
\begin{align}
E[R_s]\leq
\frac{1}{2(L-1)}\sum_{\ell=1, \ell\neq
i}^{L}\left(\log\left(\frac{1}{1+\frac{\sigma^2_
{h_{\ell,r}}}{\sigma^2_{h_{i,r}}}}+\frac{P\sigma^2_{
h_{i,r}}}{N_0}\right)+\log\left(\frac{1}{1+\frac{\sigma^2_
{h_{i,r}}}{\sigma^2_{h_{\ell,r}}}}+\frac{P\sigma^2_
{h_{\ell,r}}}{N_0}\right)\right),\label{eq:Rsavg}
\end{align}\par}
\hrulefill
\end{figure*}

Similarly, the average sum rate for the pairing scheme in
\cite{Ong:2010} can be written as
\vspace{-12pt}
{\fontsize{11}{13.2}\selectfont
\begin{equation}\label{eq:Rsongavg}
E[R_s]\leq
\frac{1}{2(L-1)}\sum_{\ell=2}^{L}\left(\log\left(\frac{1}{1+\frac{\sigma^2_
{h_{\ell,r}}}{\sigma^2_{h_{\ell-1,r}}}}+\frac{P\sigma^2_{
h_{\ell-1,r}}}{N_0}\right)+\log\left(\frac{1}{1+\frac{\sigma^2_
{h_{\ell-1,r}}}{\sigma^2_{h_{\ell,r}}}}+\frac{P\sigma^2_
{h_{\ell,r}}}{N_0}\right)\right),
\end{equation}\par}

\noindent and the average sum rate for the pairing scheme in
\cite{Noori:2012} can be written as
{\fontsize{11}{13.2}\selectfont
\begin{align}\label{eq:Rsnooriavg}
E[R_s]&\leq \frac{1}{2(L-1)}\sum_{\ell=2}^{\lfloor
L/2\rfloor+1}\left(\log\left(\frac{1}{1+\frac{\sigma^2_
{h_{\ell-1,r}}}{\sigma^2_{h_{L-\ell+2,r}}}}+\frac{P\sigma^2_{
h_{L-\ell+2,r}}}{N_0}\right)+\log\left(\frac{1}{1+\frac{\sigma^2_
{h_{L-\ell+2,r}}}{\sigma^2_{h_{\ell-1,r}}}}+\frac{P\sigma^2_
{h_{\ell-1,r}}}{N_0}\right)\right)\nonumber\\
&+\sum_{\ell=\lfloor
L/2\rfloor+2}^{L}\left(\log\left(\frac{1}{1+\frac{\sigma^2_
{h_{\ell,r}}}{\sigma^2_{h_{L-\ell+2,r}}}}+\frac{P\sigma^2_{
h_{L-\ell+2,r}}}{N_0}\right)+\log\left(\frac{1}{1+\frac{\sigma^2_
{h_{L-\ell+2,r}}}{\sigma^2_{h_{\ell,r}}}}+\frac{P\sigma^2_
{h_{\ell,r}}}{N_0}\right)\right).
\end{align}\par}



\eqref{eq:Rsavg}$-$\eqref{eq:Rsnooriavg} provide upper bounds on the actual average sum rate and they
allow an analytical comparison of the proposed and existing pairing
schemes. The main results are summarized in Propositions 4$-$6. Note that similar to the case of common rate, in Section \ref{sec:results}, the actual expression for
the instantaneous sum rate in \eqref{eq:Rs} is averaged over a large
number of channel realizations to validate the insights presented in
the Propositions 4$-$6.

\begin{proposition}\label{prop:eq_sum}
The average sum rate of the proposed pairing scheme and the
pairing schemes in \cite{Ong:2010} and \cite{Noori:2012} are the
same for the equal average channel gain scenario.
\end{proposition}

\begin{proof}
See Appendix D.
\end{proof}

\begin{proposition}\label{prop:uneq_sum}
The average sum rate of the proposed pairing scheme is larger than
that of the pairing schemes in \cite{Ong:2010} and
\cite{Noori:2012} for the unequal average channel gain scenario.
\end{proposition}

\begin{proof}
See Appendix D.
\end{proof}

\begin{proposition}\label{prop:var_sum}
The average sum rate of the proposed pairing scheme is larger than
that of the pairing schemes in \cite{Ong:2010} and
\cite{Noori:2012} for the variable average channel gain scenario.
\end{proposition}

\begin{proof}
See Appendix D.
\end{proof}

\section{Error Performance Analysis}\label{error}
%
In this section, we characterize the error performance of a FDF MWRN with the new pairing scheme. We provide the analytical derivations for $M$-QAM modulation, which is a $2$ dimensional lattice code and is widely used in practical wireless communication systems. 

\subsection{System Model}

In the $M$-QAM modulated FDF MWRN system, during a certain time frame,
in the $t_s=(\ell-1)^{th}$ time slot, the $i^{th}$ user and the
$\ell^{th}$ user transmit their messages $\textbf{W}_i$ and
$\textbf{W}_\ell$ which are $M$-QAM modulated to
$\textbf{X}_i=\{X_i^{1},X_i^{2},...,X_i^{T}\}$ and
$\textbf{X}_\ell=\{X_{\ell}^{1},X_{\ell}^{2},...,X_{\ell}^{T}\}$,
respectively, where $X_i^t,X_{\ell}^{t}=a+jb$ and $a,b\in\{\pm1,\pm3,...,\pm(\sqrt{M}-1)\}$. The relay
receives the signal $\textbf{R}_{i,\ell}$ (see \eqref{eq:received})
and decodes it using ML criterion~\cite{MJu:2010} and
obtains an estimate $\hat{\textbf{V}}_{i,\ell}$ of the network
coded symbol
$\textbf{V}_{i,\ell}=(\textbf{W}_{i}+\textbf{W}_{\ell})\mod M$ as in \cite{Zhang-2006,Ronald:2013}. The relay then broadcasts the estimated network coded signal after $M$-QAM
modulation, which is given as $\textbf{Z}_{i,\ell}$. The $j^{th}$
($j\in[1,L]$) user receives $\textbf{Y}_{i,\ell}$ (see
\eqref{eq:dfdetect}) and detects the received signal through ML
criterion~\cite{MJu:2010} to obtain the estimate
$\dhat{\textbf{V}}_{i,\ell}$. After decoding all the network coded
messages, each user performs message extraction. For the common user
($i^{th}$ user), this message extraction involves subtracting its own
message $\textbf{W}_{i}$ from the
network coded messages $\dhat{\textbf{V}}_{i,\ell}$ and then performing the modulo-$M$ operation. The process can be shown as
{\fontsize{10.4}{12.48}\selectfont
\begin{equation}\label{eq:commonextract}
\hat{\textbf{W}}_{\ell}=(\dhat{\textbf{V}}_{i,\ell}-
\textbf{W}_{i}+M)\mod M,
\hat{\textbf{W}}_{\ell+1}=(\dhat{\textbf{V}}_{i,\ell+1}-
\textbf{W}_{i}+M) \mod M,...,\hat{\textbf{W}}_{L}=(\dhat{\textbf{V}}_{i,L}-\textbf{W}_{i}+M)\mod M.
\end{equation}\par}

For other users, the message extraction process can be shown as
{\fontsize{10.4}{12.48}\selectfont
\begin{equation}\label{eq:otherextract}
\hat{\textbf{W}}_{i}=(\dhat{\textbf{V}}_{i,\ell}-
\textbf{W}_{\ell}+M)\mod M,\hat{\textbf{W}}_{\ell+1}=(\dhat{\textbf{V}}_{i,\ell+1}-
\hat{\textbf{W}}_{i}+M)\mod M,...,\hat{\textbf{W}}_{L}=(\dhat{\textbf{V}}_{i,L}-
\hat{\textbf{W}}_{i}+M)\mod M.
\end{equation}\par}

%
\normalsize
\subsection{SER Analysis for the Proposed Pairing
Scheme}\label{sec:proposed} In this subsection, we investigate the
error performance of a FDF MWRN with the proposed pairing scheme.
Unlike the pairing schemes in \cite{Ong:2010} and \cite{Noori:2012},
the error performance of all the users is not the same for the
proposed pairing scheme. Hence, we need to obtain separate
expressions for the error probabilities at the common user ($i^{th}$
user) and other users ($\ell^{th}$ user).

First, we obtain the probability of incorrectly decoding a network
coded message at the common user and the other users. Since, any $M$-QAM signal with square constellation (i.e., $\sqrt{M}\in\mathbb{Z}$) can be decomposed to two $\sqrt{M}$-PAM signals \cite{Alouini:2000}, the network coded signal from a linear combination of two $M$-QAM signals can be decomposed to a network coded signal from two $\sqrt{M}$-PAM signals. Thus, we can obtain the probability that the $i^{th}$ (common) user incorrectly decodes the
network coded message involving its own message and the $m^{th}$ user's message, as:
\begin{align}\label{eq:df2waycommon}
P_{FDF}(i,m)&=1-\left(1-P_{\sqrt{M}-PAM,NC}(i,m)\right)^2,
\end{align}

\noindent where $P_{\sqrt{M}-PAM,NC}(i,m)$ is the probability of incorrectly decoding a network coded message resulting from the sum of two $\sqrt{M}$-PAM signals from the $i^{th}$ and the $m^{th}$ user and is derived in Appendix E.

%
%

%

Similarly, The probability that the $\ell^{th}$ (other) user incorrectly
decodes the network coded message involving the $i^{th}$ user's message and its own message or other user's messages is given as:
\begin{align}\label{eq:df2wayother}
P_{FDF}(\ell,m)=\left\{
\begin{array}{ll}1-\left(1-P_{\sqrt{M}-PAM,NC}(\ell,m)\right)^2
& \mbox{$m=i$}\\
1-\left(1-P_{\sqrt{M}-PAM,NC}(i,m)\right)^2
& \mbox{$m\in[1,L],m\neq i,\ell$}.
\end{array}
\right.
\end{align}

\noindent where $P_{\sqrt{M}-PAM,NC}(\ell,m)$ is the probability of incorrectly decoding a network coded message, i.e., the sum of two $\sqrt{M}$-PAM signals of the $\ell^{th}$ and the $m^{th}$ user and can be obtained from Appendix E.
Using \eqref{eq:df2waycommon} and \eqref{eq:df2wayother}, the
average SER at the common user and the other users can be derived
using the technique proposed in \cite{Shama:2012}. The result is
summarized in the following Theorem.

\begin{theorem}\label{th:aber}
For the proposed pairing scheme in a FDF MWRN, the average SER at
the $i^{th}$ (common) user is given by:
\begin{equation}\label{eq:abercommon}
P_{i,avg}=\frac{1}{L-1}{\sum_{m=1,m\neq i}^{L}P_{FDF}(i,m)},
\end{equation}

\noindent and the average SER at the $\ell^{th}$ (other) users is given by:
\begin{equation}\label{eq:aberother}
P_{\ell,avg}=\frac{1}{L-1}\left({\sum_{m=1,m\neq
i,\ell}^{L}P_{FDF}(\ell,m)+(L-1)P_{FDF}(\ell,i)}\right).
\end{equation}
\end{theorem}

\begin{proof}
See Appendix F.
\end{proof}

\begin{remark}
From Theorem \ref{th:aber}, it can be identified that the average SER at the other ($\ell^{th}$) users is at least twice compared to the average SER at the common ($i^{th}$) user. This can be intuitively explained from the fact that the $i^{th}$ user needs to correctly decode only one network coded message (${V}_{i,m}$) to correctly decode the $m^{th}$ user's message. However, the $\ell^{th}$ user needs to correctly decode two network coded messages ($V_{i,m}$ and $V_{i,\ell}$) to correctly decode the $m^{th}$ user's message. Thus, the average SER at the other users would at least be twice compared to that at the common user. 
\end{remark}

Using Theorem \ref{th:aber} and the average SER result for the pairing scheme in \cite{Ong:2010}, we can compare the performance of the proposed and the existing pairing schemes. Note that the error performance of the pairing scheme in
\cite{Noori:2012} would be the same as the pairing scheme in \cite{Ong:2010}, as the basic pairing process is the same for both these schemes and only the pairing orders are different. The main results are summarized in Propositions 7$-$9.

\begin{proposition}\label{prop:eq_aber}
The average SER of an $L$-user FDF MWRN with the proposed pairing
scheme is lower than the pairing scheme in \cite{Ong:2010} by a
factor of $\frac{L}{2}$ for the common user and a
factor of approximately $\frac{L}{4}$ for other users under the equal average channel gain
scenario.
\end{proposition}

\begin{proof}
See Appendix G.
\end{proof}

\begin{proposition}\label{prop:uneq_aber}
The average SER of an $L$-user FDF MWRN with the proposed pairing
scheme is always lower than the pairing scheme in \cite{Ong:2010}
for all users under the unequal average channel gain scenario.
\end{proposition}

\begin{proof}
See Appendix G. \end{proof}
%
%

\begin{proposition}\label{prop:var_aber}
The average SER of an $L$-user FDF MWRN with the proposed pairing
scheme is always lower than the pairing scheme in \cite{Ong:2010}
for all users under the variable average channel gain scenario.
\end{proposition}

\begin{proof}
See Appendix G. \end{proof}

From Propositions \ref{prop:eq_aber}-\ref{prop:var_aber}, it is
clear that choosing the user with the best average channel gain as
the common user reduces the average SER of the FDF MWRN.

\section{Results}\label{sec:results}
In this section, we provide numerical results to verify the insights
provided in Propositions \ref{prop:eq_cap}$-$\ref{prop:var_sum}. We
also provide simulation results to verify Propositions
\ref{prop:eq_aber}$-$\ref{prop:var_aber}. We consider an $L=10$ user
FDF MWRN where each user transmits a packet of $T=2000$ bits and uses $16$-QAM modulation. The
power at the users, $P$ and the power at the relay, $P_r$ are
assumed to be equal and normalized to unity. The SNR per bit per
user is defined as $\frac{1}{N_0}$. Following \cite{Gayan:2013}, the
average channel gain for the $j^{th}$ user is modeled by
$\sigma^2_{h_{j,r}}=(1/(d_j/d_0))^\nu$, where $d_0$ is the reference
distance, $d_j$ is the distance between the $j^{th}$ user and the
relay which is assumed to be uniformly randomly distributed between
$0$ and $d_0$, and $\nu$ is the path loss exponent, which is assumed
to be $3$. Such a distance based channel model takes into account
large scale path loss and has been widely considered in the
literature
\cite{Zhang-2006,Foroogh:2012,Louie-2010,MChen:2010,Kramer:2005}. 
All distances, once chosen, remain constant for unequal channel gain scenario and are
randomly chosen every time frame (i.e., worst case, $T'_f=1$) for
variable channel gain scenario. Note that all the distances are the same for the equal average channel gain scenario. All results are averaged over $F=100$ time frames.


\subsection{Common Rate}
Fig. \ref{fig:Fig5} shows the common rate for the proposed
and the existing pairing schemes in an $L=10$ user FDF MWRN. All the
numerical results are obtained by averaging the instantaneous common
rates for the pairing schemes over a large number of channel
realizations. Fig. \ref{fig:su1} shows that all the pairing schemes
have the same average common rate in equal average channel gain
scenario, which verifies Proposition \ref{prop:eq_cap}. The common
rate of the proposed pairing scheme is larger than the existing
pairing schemes for the unequal average channel gain scenario in
Fig. \ref{fig:su2}. This is because, scaling the common user's power
to ensure transmission fairness decreases the ratio of the maximum
and the minimum average channel gains in \eqref{eq:Rcnew}, resulting
in a larger common rate. For variable average channel gain scenario,
we can see that the common rate for the proposed scheme is practically the same as that of the
existing pairing schemes. This verifies Propositions
\ref{prop:uneq_cap} and \ref{prop:var_cap}, respectively.

%
\begin{figure}
{\subfigure[Equal channel gain scenario]{
  \includegraphics[width=0.33\textwidth]{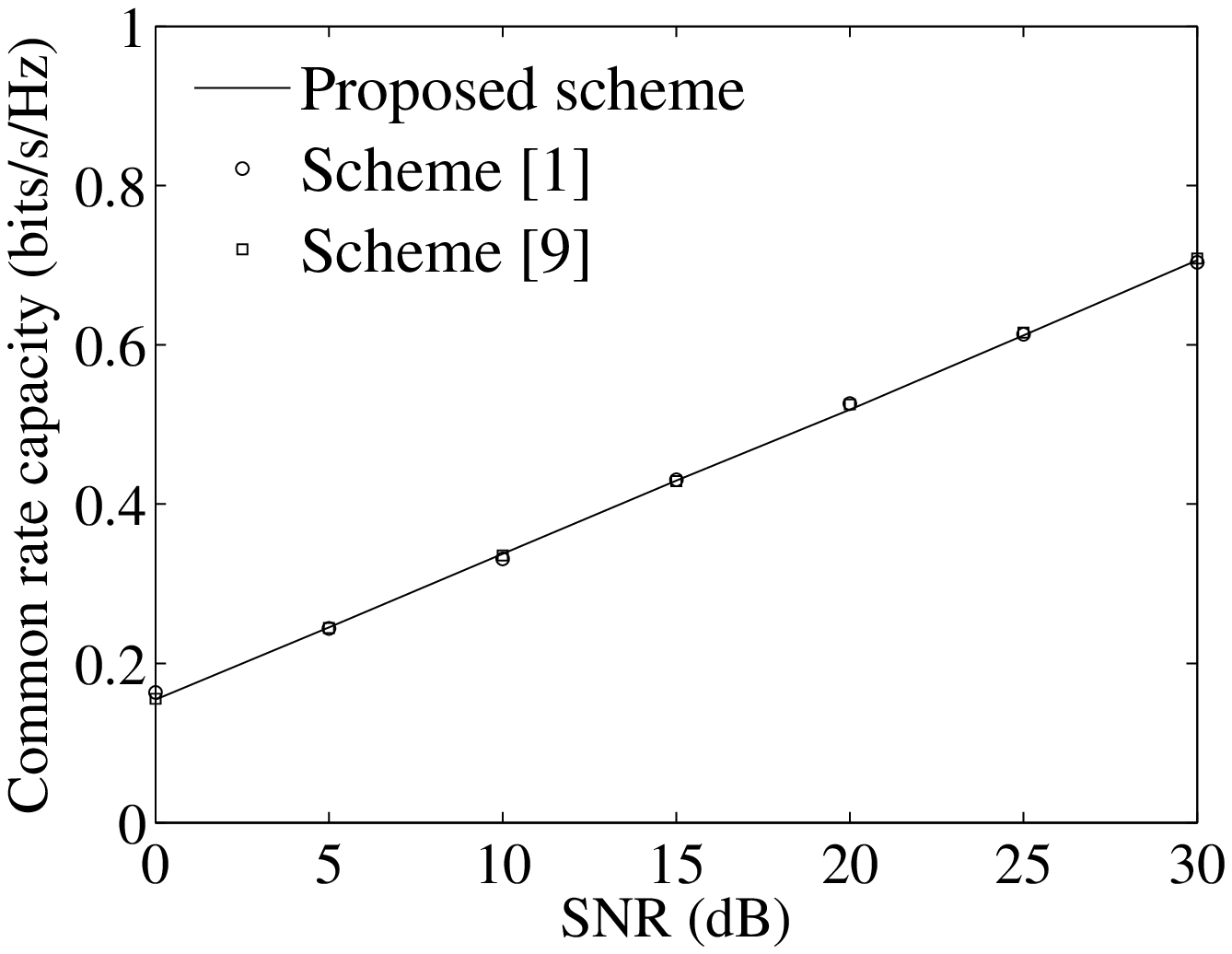}
  \label{fig:su1}}
\hspace{-20pt}\subfigure[Unequal channel gain scenario]{
  \includegraphics[width=0.33\textwidth]{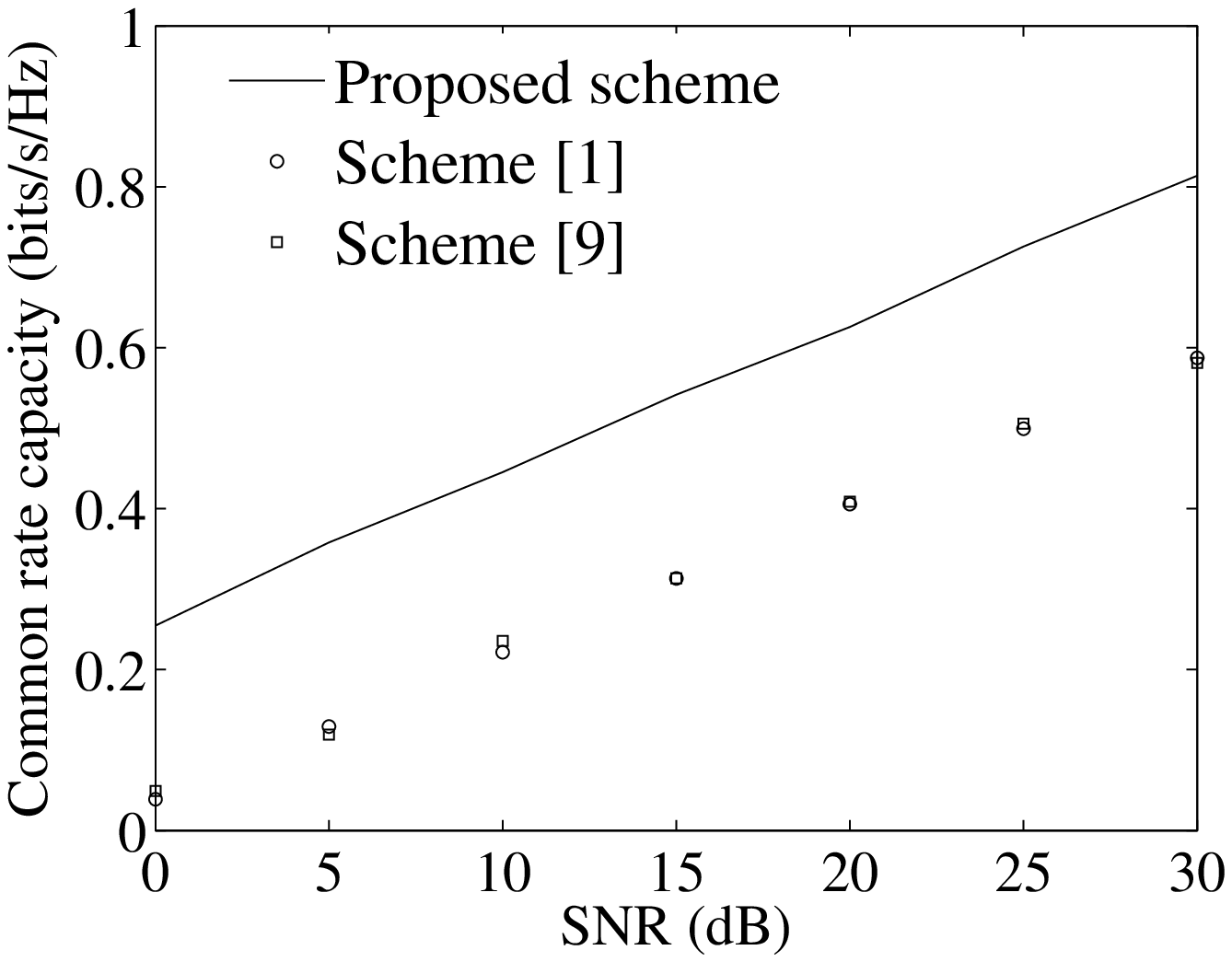}
 \label{fig:su2}}
\hspace{-20pt}\subfigure[Variable channel gain scenario]{
  \includegraphics[width=0.33\textwidth]{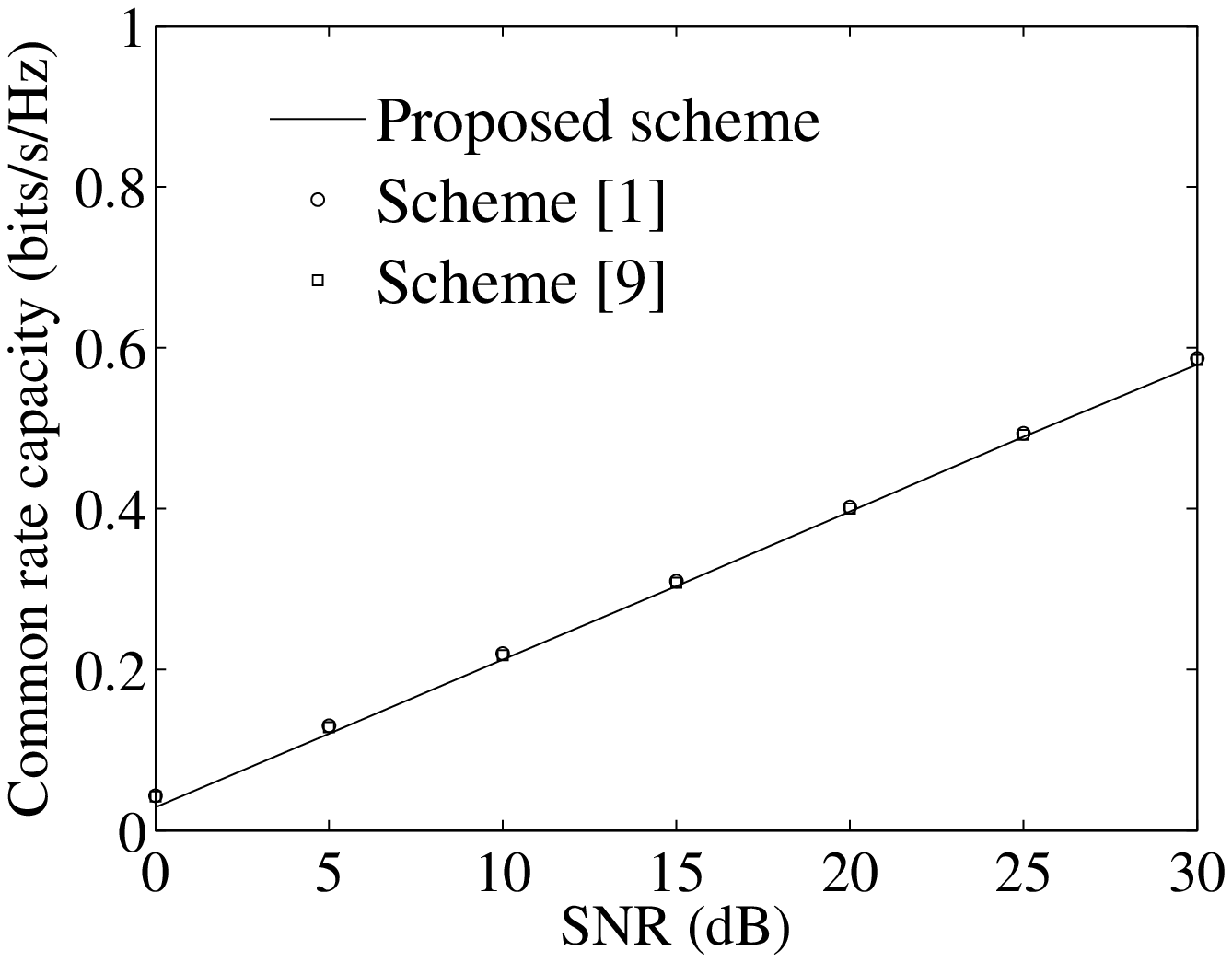}
 \label{fig:su3}}
   \caption{Common rate for a $L=10$ user FDF MWRN with different pairing schemes and different channel scenarios.}
\label{fig:Fig5}}
\end{figure}
\vspace{-15pt}
\subsection{Sum Rate}
Fig. \ref{fig:Fig6} shows the sum rate for the proposed and the
existing pairing schemes in an $L=10$ user FDF MWRN for the three
channel scenarios. All the numerical results are obtained by
averaging the instantaneous sum rates for the pairing schemes over a
large number of channel realizations. Fig. \ref{fig:su4} shows that
all the pairing schemes have the same average sum rate for equal
average channel gain scenario, which verifies Proposition \ref{prop:eq_sum}.
Similarly, Fig. \ref{fig:su5} and Fig. \ref{fig:su6} show that the
average sum rate for the proposed pairing scheme is larger than the
existing pairing schemes, which is in line with the Propositions
\ref{prop:uneq_sum} and \ref{prop:var_sum}. Intuitively, this can be explained as follows. In the proposed pairing scheme, the common user with the maximum average channel gain transmits more times than the other users. Unless all the average channel gains are equal, this results in a larger sum rate compared to the existing pairing schemes.
\begin{figure}
{\subfigure[Equal channel gain scenario]{
  \includegraphics[width=0.33\textwidth]{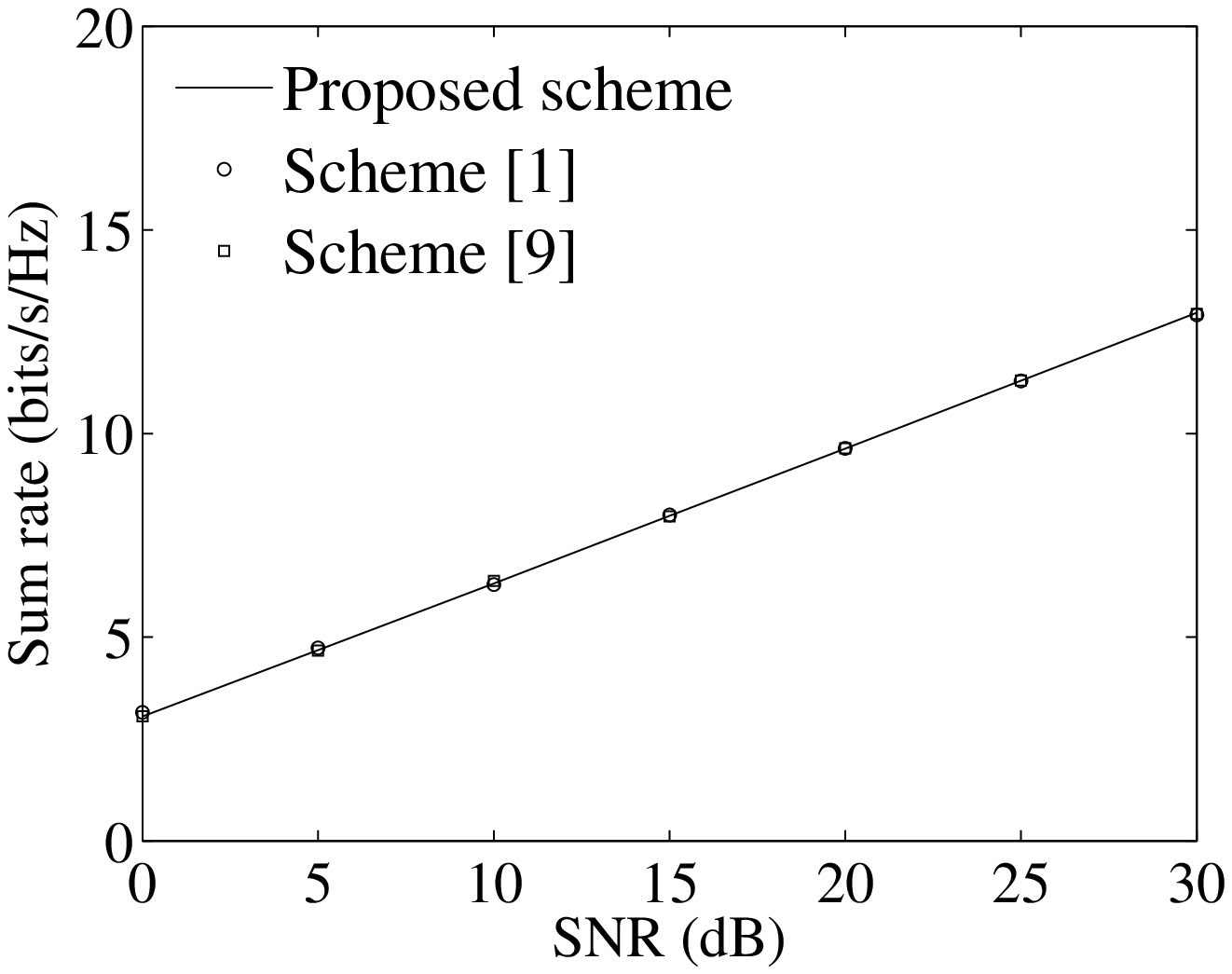}
  \label{fig:su4}}
\hspace{-20pt}\subfigure[Unequal channel gain scenario]{
  \includegraphics[width=0.33\textwidth]{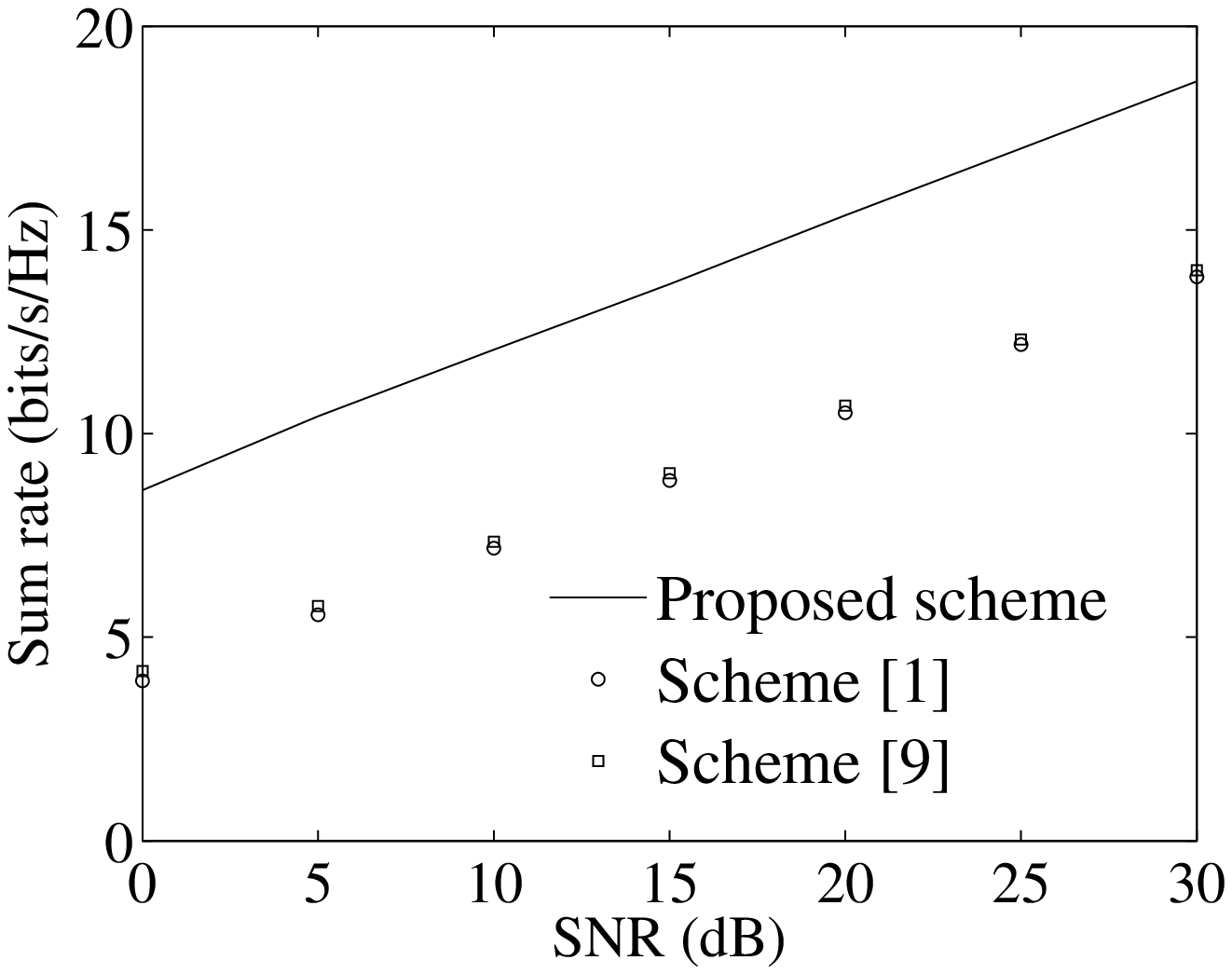}
 \label{fig:su5}}
\hspace{-20pt}\subfigure[Variable channel gain scenario]{
  \includegraphics[width=0.33\textwidth]{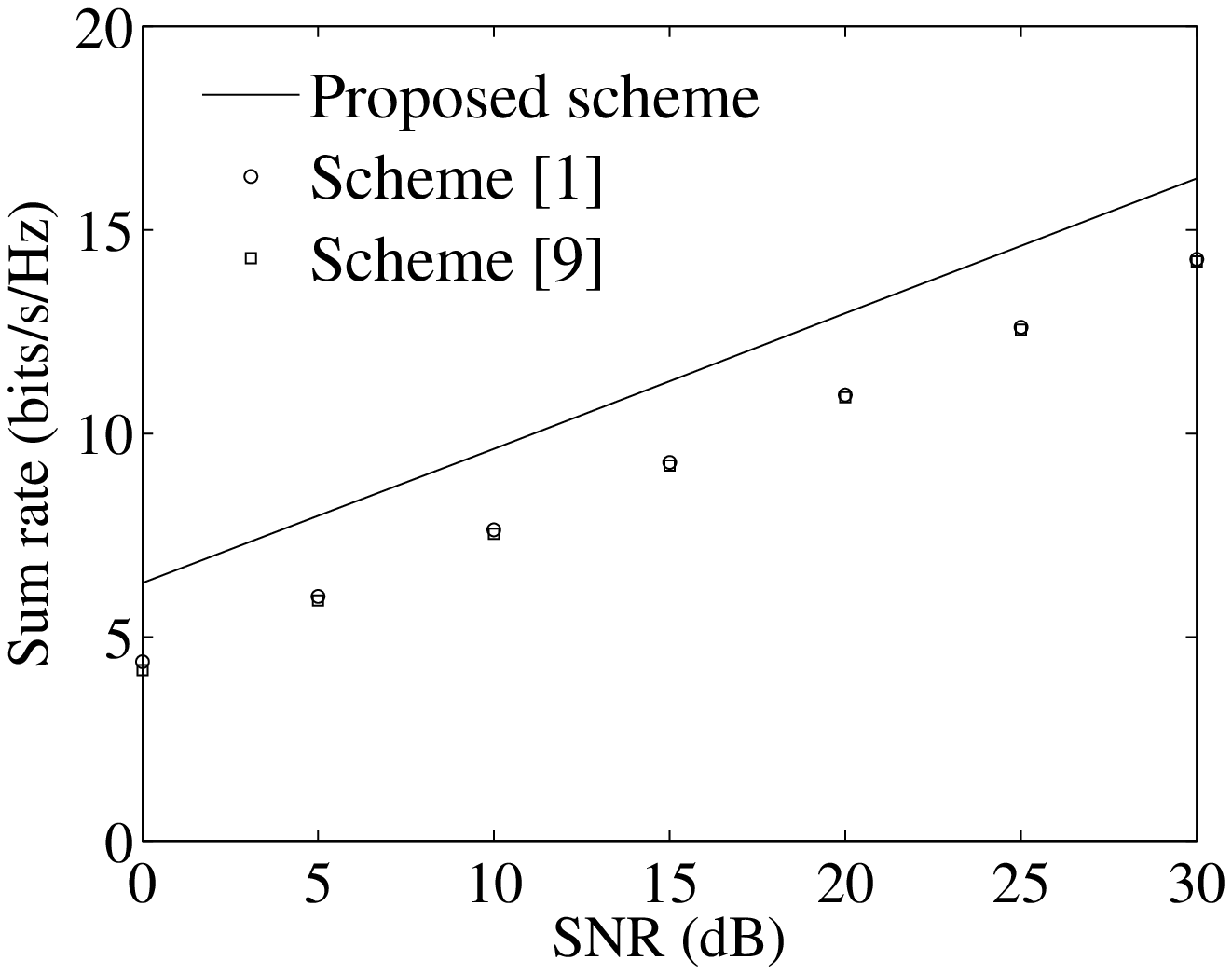}
 \label{fig:su6}}
   \caption{Sum rate for a $L=10$ user FDF MWRN with different pairing schemes and different channel scenarios.}
\label{fig:Fig6}}
\end{figure}
\subsection{Robustness of the Proposed Pairing Scheme}

To illustrate robustness of the proposed pairing scheme, we consider
two special cases of the variable average channel gain scenario,
where (i) $10\%$ of the users have distances below $0.1d_0$ (i.e.,
only a small proportion of the users are close to the relay and so,
they have good channel conditions) and (ii) $90\%$ of the users have
distances below $0.1d_0$ (i.e., a large proportion of users have good channel conditions). Fig.
\ref{fig:su7} plots the average common rate and Fig.
\ref{fig:su8} plots the average sum rate for the proposed and
existing pairing schemes. We can see from Fig. \ref{fig:su7} that
the common rate does not change much when either $10\%$ or $90\%$ of
users have good channel conditions as it depends upon the minimum average channel gain
in the system. However, we can see from Fig. \ref{fig:su8} that 
when the number of users with good channel conditions falls from
$90\%$ to $10\%$, the sum rate of the proposed scheme degrades to a
much lesser extent, compared to the existing pairing schemes. This
is because the average sum rate of the proposed pairing scheme
depends to a greater extent on the common user's average channel
gain compared to the other users' average channel gain (as evident
from \eqref{eq:Rsavg}). However, for the existing pairing schemes,
the sum rate depends on all the channel gains equally (as evident
from \eqref{eq:Rsongavg} and \eqref{eq:Rsnooriavg}) and degrades to a greater extent.
This illustrates the
robustness of the proposed pairing scheme.

\begin{figure}
{\subfigure[Common rate]{
  \includegraphics[width=0.44\textwidth]{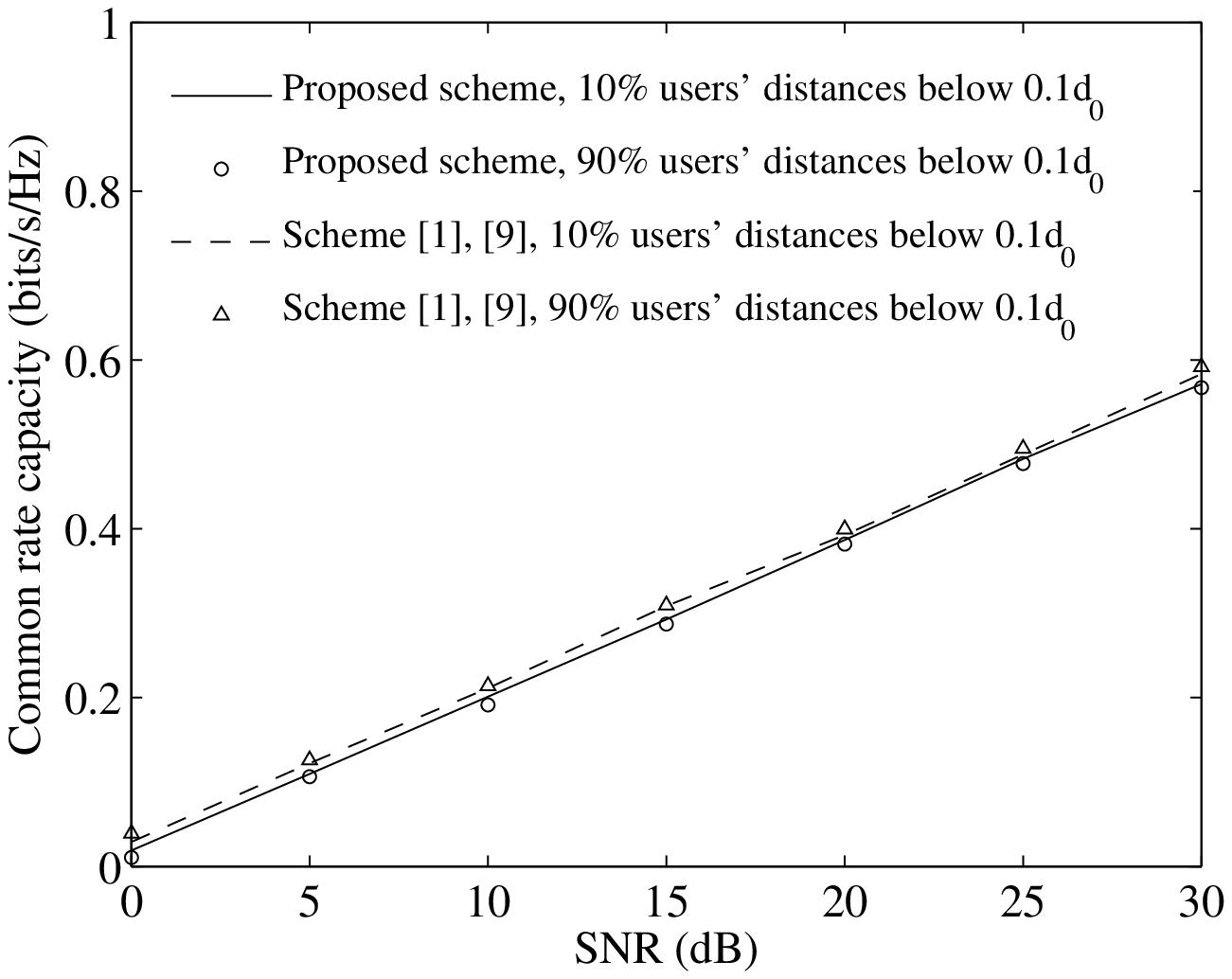}
  \label{fig:su7}}
\hspace{-20pt}\subfigure[Sum rate]{
  \includegraphics[width=0.44\textwidth]{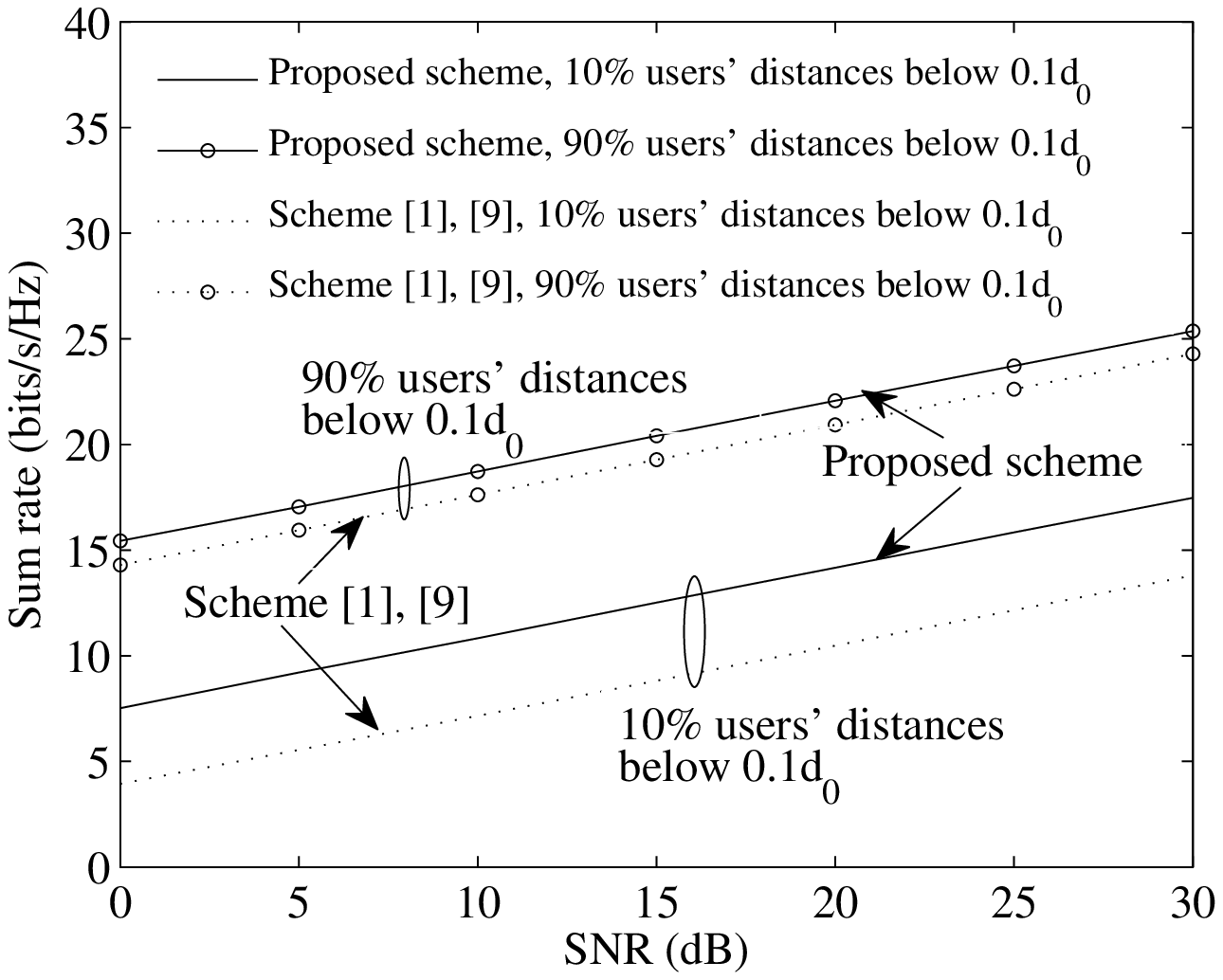}
 \label{fig:su8}}
  \caption{Common rate and sum rate of an $L=10$ user FDF
  MWRN when 10\% and 90\% users have distances below $0.1d_0$.}
\label{fig:Fig8}}
\end{figure}
\vspace{-15pt}
\subsection{Average SER}

Figures \ref{fig:su9}, \ref{fig:su10} and \ref{fig:su11} plot the
average SER of the proposed and the existing pairing schemes in an
$L=10$ user FDF MWRN for equal channel gain scenario (Fig.
\ref{fig:su9}), unequal channel gain scenario (Fig. \ref{fig:su10})
and variable channel gain scenario (Fig. \ref{fig:su11}). We can see
from all the figures that the simulation results match perfectly
with the analytical results at mid to high SNRs. 
This verifies the accuracy of Theorem \ref{th:aber}. Note
that the existing pairing schemes in \cite{Ong:2010} and
\cite{Noori:2012} have the same average SER. So, only the results
for pairing scheme in \cite{Ong:2010} have been shown in the above
figures. Figures \ref{fig:su9}, \ref{fig:su10} and \ref{fig:su11}
show that the proposed pairing scheme outperforms the existing
pairing schemes, in terms of average SER, which verifies
Propositions \ref{prop:eq_aber}-\ref{prop:var_aber}. In addition,
Fig. \ref{fig:su9} shows that the average SER at the common user and
other users are $5$ times and nearly $2.5$ times less than that of
the existing pairing schemes. This verifies the insight presented by Remark 1 and Proposition
\ref{prop:eq_aber}.

Fig. \ref{fig:su12} plots the average SER of the proposed and the
existing pairing schemes for the special cases of the variable
average channel gain scenario when (i) $10\%$ of the users have
distances below $0.1d_0$ and (ii) $90\%$ of the users have distances
below $0.1d_0$. The figure shows that the average SER for the
existing pairing schemes worsens by a larger extent compared to
that of the proposed scheme with the degradation in the users'
channel conditions. For the proposed pairing scheme, when the number
of users with good channel conditions increases from $10\%$ to
$90\%$, the average SER at other users improve significantly and
approaches the average SER at the common user. This is because the
average SER at the $\ell^{th}$ user depends not only on its own
channel conditions, but also the channel conditions of the common
($i^{th}$) user and the $m^{th}$ user (see \eqref{eq:aberother}).
This improvement in the overall channel conditions results in
improvement in the average SER, which illustrates the superiority of
the proposed pairing scheme.

\begin{figure}
{\subfigure[Equal channel gain scenario]{
  \includegraphics[width=0.44\textwidth]{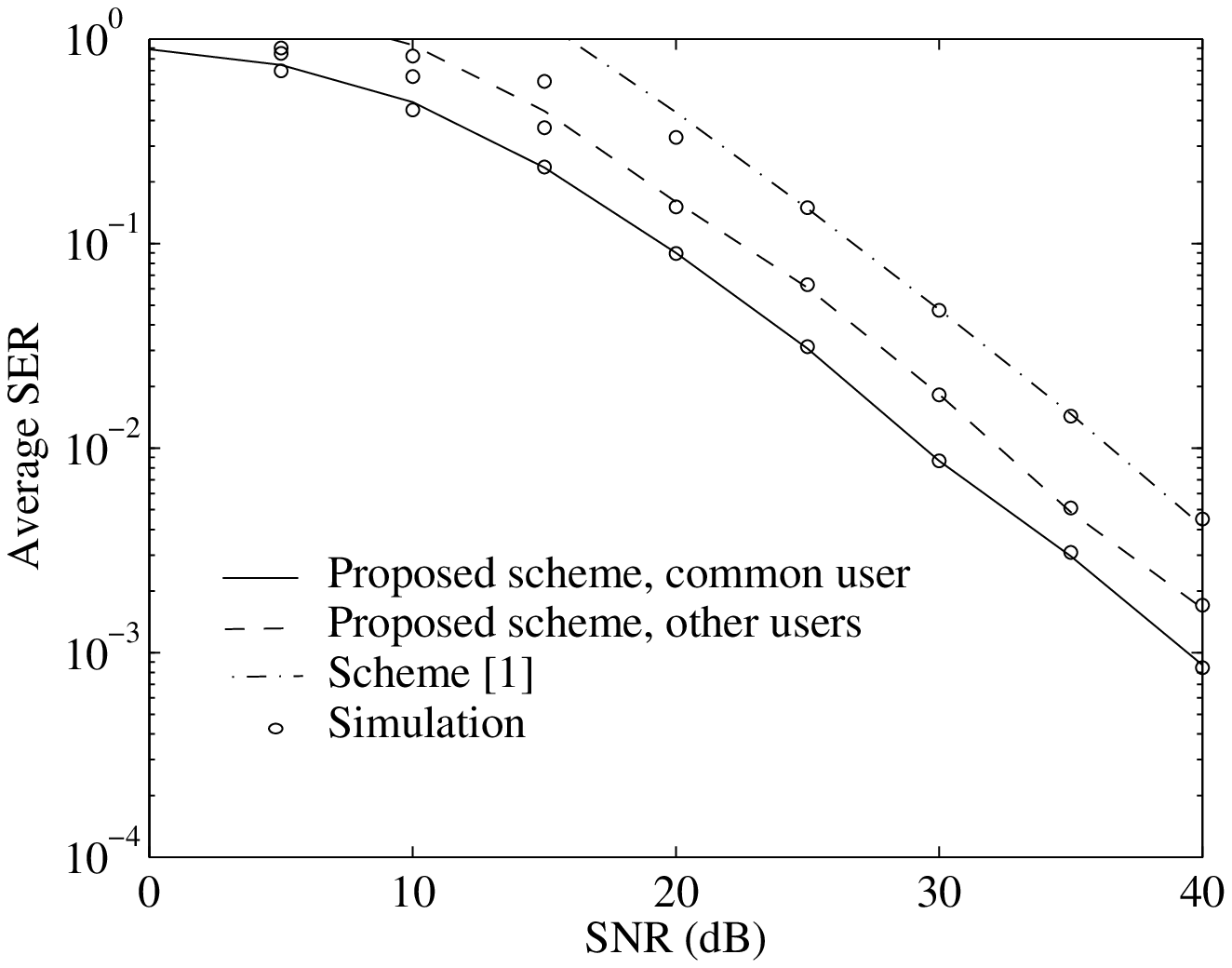}
  \label{fig:su9}}
\hspace{-20pt}\subfigure[Unequal channel gain scenario]{
  \includegraphics[width=0.44\textwidth]{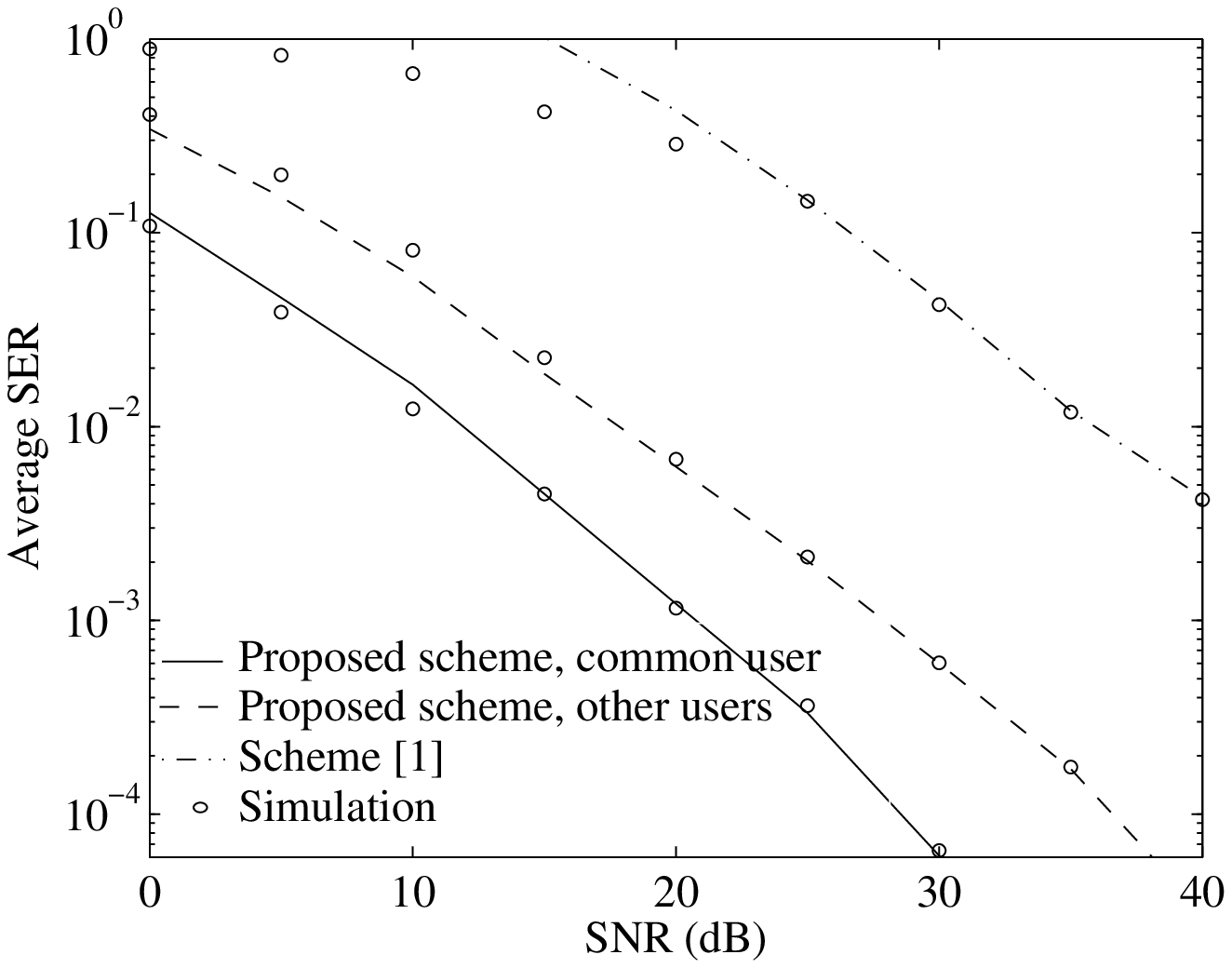}
 \label{fig:su10}}
  \caption{Average SER for equal and unequal average channel gains in an $L=10$ user FDF
  MWRN with different pairing schemes.}
\label{fig:Fig2}}
\vspace{-10pt}
\end{figure}
\begin{figure}
{\subfigure[Variable channel gain scenario]{
  \includegraphics[width=0.44\textwidth]{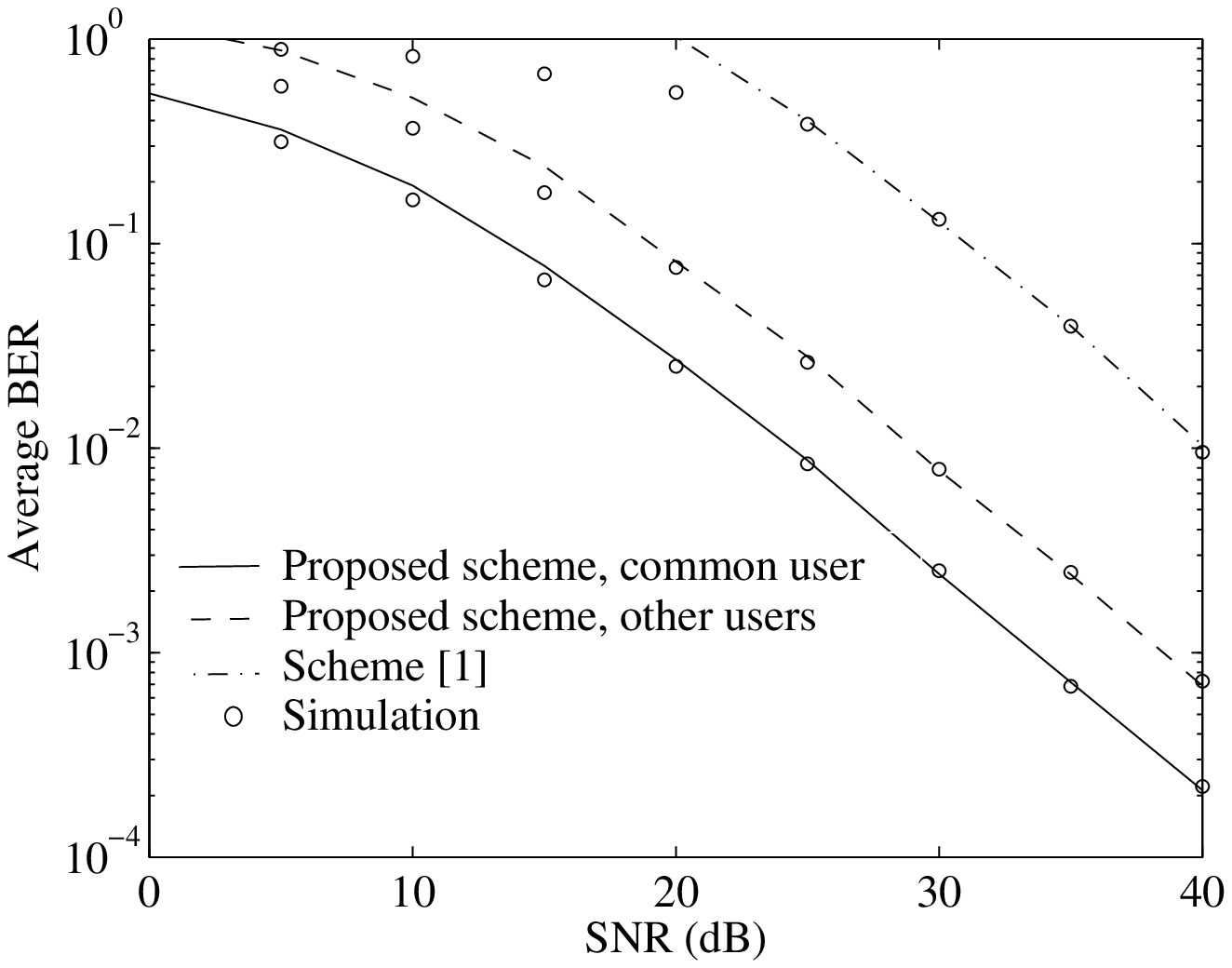}
  \label{fig:su11}}
\hspace{-20pt}\subfigure[Variable channel gain scenario where $10\%$
and $90\%$ users, respectively, have distances below $0.1d_0$.]{
  \includegraphics[width=0.44\textwidth]{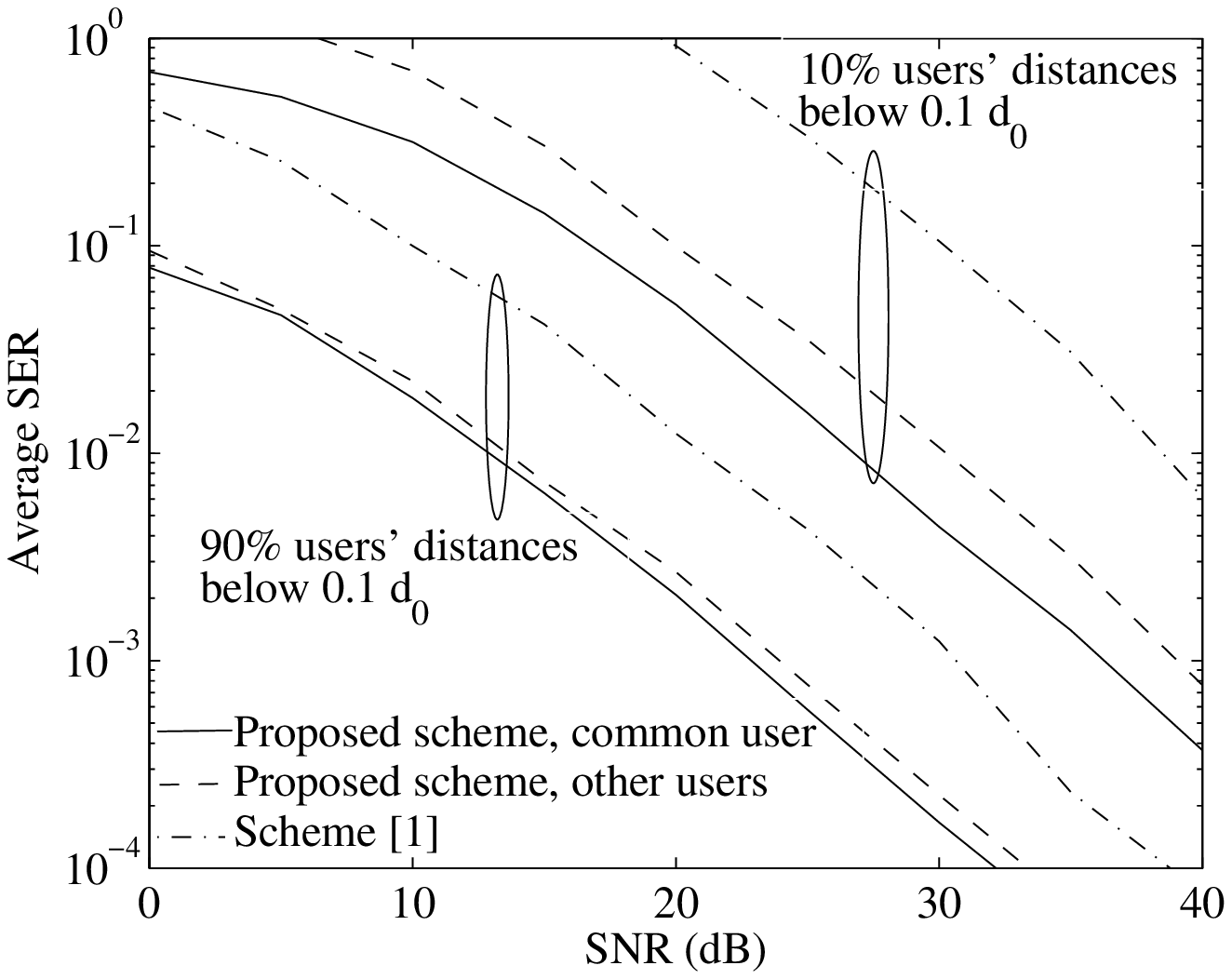}
 \label{fig:su12}}
  \caption{Average SER for variable average channel gains in an $L=10$ user FDF
  MWRN with different pairing schemes.}
\label{fig:Fig3}}
\vspace{-10pt}
\end{figure}
\vspace{-15pt}
\section{Conclusions} \label{conclusion}
In this paper, we have proposed a novel user pairing scheme in a FDF
MWRN. We have derived the upper bound on the average common rate (Theorem \ref{th:cap})
and the average sum rate (Theorem \ref{th:sum}) and the asymptotic average SER (Theorem \ref{th:aber}) for the proposed pairing scheme. We have analyzed the results in Theorems
\ref{th:cap}-\ref{th:aber} to compare the performance of the proposed scheme with
existing pairing schemes under different channel scenarios. The main
insights are summarized in Propositions
\ref{prop:eq_cap}-\ref{prop:var_aber}. Our
analysis shows that the proposed pairing scheme improves the
aforementioned performance metrics compared to that of the existing
pairing schemes for different channel conditions. 
\vspace{-20pt}
\section*{Appendix A\\ \vspace{-10pt} Proof of Theorem
\ref{th:cap}}\label{proof:cap}

In the proposed pairing scheme, the $i^{th}$ and the $\ell^{th}$
user transmit simultaneously in the $(\ell-1)^{th}$ time slot in the
multiple access phase. Also, in the broadcast phase, in the
$(\ell-1)^{th}$ time slot, the relay broadcasts the decoded network
coded message to all the users. For the \textit{multiple access
phase}, the optimum values of $\alpha$ and $\beta_j$ in
\eqref{eq:snrlatticenp} and \eqref{eq:snrlatreceived}, respectively,
are obtained by setting $\frac{dn}{d\alpha}=0$ and
$\frac{dn'}{d\beta_j}=0$, where $n$ and $n'$ are given in
\eqref{eq:scaling} and \eqref{eq:rec_scaled}, respectively. From
this, we obtain $\alpha=\frac{P\mid h_{i,r}\mid^2+P\mid
h_{\ell,r}\mid^2}{P\mid h_{i,r}\mid^2+P\mid h_{\ell,r}\mid^2+N_0}$
and $\beta_j=\frac{P_r\mid h_{j,r}\mid^2}{P_r\mid
h_{j,r}\mid^2+N_0}$. Substituting these values in
\eqref{eq:snrlatticenp} and \eqref{eq:snrlatreceived}, \eqref{eq:Rm}
and \eqref{eq:Rb} can be derived following the steps in
\cite{Nam:2011} and \cite{Noori:2012}, which are summarized as
follows. First, we assume that there exists a rate
$\bar{R}<R_{M,\ell-1}$ for which Pr$(n\notin\mathcal{V})$ (see
\eqref{eq:scaling}) is upper bounded by $e^{-N(E_p(\mu))}$, where
$E_p$ is the Poltyrev exponent,
$\mu=2^{2(R_{M,\ell-1}-\bar{R})}-O_N(1)$ \cite{Nam:2011} is the
volume to noise ratio of the lattice $\Lambda$ with respect to the
noise $n$, $O_N(1)$ indicates that the difference between $\mu$ and
$2^{2(R_{M,\ell-1}-\bar{R})}$ is a first degree function of $N$ and
$\Lambda$ is Poltyrev-good \cite{Nam:2011}. Then calculating $\mu$
and comparing with $2^{2(R_{M,\ell-1}-\bar{R})}-O_N(1)$ gives
\eqref{eq:Rm} in Theorem~\ref{th:cap}. For the \textit{broadcast
phase}, \eqref{eq:Rb} in Theorem~\ref{th:cap} can be obtained from
the point to point channel of the users. The details are
omitted here for the sake of brevity. This completes the proof.

\section*{Appendix B\\ \vspace{-11pt} Proof of Propositions
\ref{prop:eq_cap}$-$\ref{prop:var_cap}}\label{app:cap}

\noindent{\textit{Proof of Proposition 1:}} For the equal average channel gain scenario, $\sigma^2_{
h_{i,r}}=\sigma^2_{h_{\ell,r}}=\sigma^2_{
h_{\ell-1,r}}=\sigma^2_{h_{L-\ell+2,r}}$. Thus, the average common
rate expressed by \eqref{eq:Rcong}, \eqref{eq:Rcnoori} and
\eqref{eq:Rcnew} becomes the same for all the three pairing schemes.
This proves Proposition \ref{prop:eq_cap}.

\noindent{\textit{Proof of Proposition 2:}} For unequal average
channel gain scenario, as explained in Section~\ref{sec:pairing},
the transmit power of the $i^{th}$ user needs to be scaled by
$(L-1)$ to ensure transmission fairness. As a result, $\mid
h_{i,r}\mid^2$ can be replaced by $\frac{\mid h_{i,r}\mid^2}{L-1}$
in \eqref{eq:snrlatticenp}. In addition, for a fair comparison with
the existing pairing schemes, the transmit power $P$ in the proposed
scheme, needs to be multiplied by a factor $(2L-2)$. This is because
in the proposed pairing scheme, the common user transmits $(L-1)$
times with power $\frac{P}{L-1}$ and the other $(L-1)$ users
transmit once with power $P$. Hence, the average power per user
becomes $P$. However, for the existing pairing schemes, the average
power per user is $\frac{2L-2}{L}P$. Overall, \eqref{eq:Rcnew} can
be modified by scaling $\sigma^2_{h_{i,r}}$ with $L-1$ and replacing
$P$ with $(2L-2){P}$. Thus, the average common rate in
\eqref{eq:Rcnew} is
{\fontsize{11}{13.2}\selectfont
\begin{equation}\label{eq:Rcunequal}
E[R_{c}]\leq \frac{1}{2(L-1)}\log\left(\min\left(\frac{1}{1+
\frac{(L-1)\sigma_{ h_{\ell,r}}^2}{\sigma_
{h_{i,r}}^2}}+\frac{(2L-2)P\sigma_{h_{i,r}}^2}{(L-1)N_0},\frac{1}{1+\frac{\sigma_{
h_{i,r}}^2}{(L-1)\sigma_{
h_{\ell,r}}^2}}+\frac{(2L-2)P\sigma_{h_{\ell,r}}^2}{N_0}\right)\right).
\end{equation}\par}

\noindent We consider two cases:
\begin{itemize}
\item \textit{case 1}: $\sigma^2_{h_{i,r}}>(L-1)\sigma^2_{h_{\ell,r}}$. In this case, the second quantity in the right hand side of
\eqref{eq:Rcunequal} will be the minimum. Then, comparing
\eqref{eq:Rcunequal} and \eqref{eq:Rcong} shows that
$\frac{(2L-2)P\sigma^2_{h_{\ell,r}}}{N_0}>\frac{P\sigma^2_{h_{\ell,r}}}{N_0}$.
Thus, the average common rate for scheme \cite{Ong:2010} will be
smaller than that for the proposed pairing scheme, when
$\sigma^2_{h_{i,r}}<(L-1)\sigma^2_{h_{\ell-1,r}}$. Similarly, it can
be shown that for the pairing scheme in \cite{Noori:2012}, the
average common rate is smaller than that for the proposed scheme for
$\sigma^2_{h_{i,r}}<(L-1)\sigma^2_{h_{L-\ell+2,r}}$.

\item \textit{case 2}: $\sigma^2_{h_{i,r}}<(L-1)\sigma^2_{h_{\ell,r}}$. In this case, the first quantity in the right hand side of
\eqref{eq:Rcunequal} will be the minimum. Then comparing
\eqref{eq:Rcunequal} and \eqref{eq:Rcong} shows that the common rate
of scheme \cite{Ong:2010} will be smaller than that of the proposed
pairing scheme, when
$\sigma^2_{h_{i,r}}>(L-1)\sigma^2_{h_{\ell-1,r}}$. Similarly, it can
be shown that for the pairing scheme in \cite{Noori:2012}, the
average common rate is smaller than that of the proposed
scheme for $\sigma^2_{h_{i,r}}>(L-1)\sigma^2_{h_{L-\ell+2,r}}$.
\end{itemize}

Combining the result from the two cases, the proposed pairing scheme will have
larger average common rate compared to the two other
pairing schemes, which proves Proposition \ref{prop:uneq_cap}.

\noindent{\textit{Proof of Proposition 3:}} For the variable channel
gain scenario, $\sigma^2_{h_{i,r}}$ in \eqref{eq:Rcnew} is the
largest average channel gain in the system. Thus, from
\eqref{eq:Rcnew}, it can be shown that $\frac{\sigma_{
h_{i,r}}^2}{\sigma_{ h_{\ell,r}}^2}>\frac{\sigma_{
h_{\ell,r}}^2}{\sigma_{ h_{i,r}}^2}$ and the second quantity in the
right hand side of the inequality in \eqref{eq:Rcnew} is the
minimum. Then comparing \eqref{eq:Rcnew} and \eqref{eq:Rcong} would
show that $\frac{\sigma_{ h_{\ell-1,r}}^2}{\sigma_
{h_{\ell,r}}^2}\leq \frac{\sigma_{ h_{i,r}}^2}{\sigma_
{h_{\ell,r}}^2}$. 
Similarly, from \eqref{eq:Rcnew} and \eqref{eq:Rcnoori}, it can be
shown that $\frac{\sigma_{ h_{L-\ell+2,r}}^2}{\sigma_
{h_{\ell-1,r}}^2}\leq \frac{\sigma_{ h_{i,r}}^2}{\sigma_
{h_{\ell,r}}^2}$ and $\frac{\sigma_{ h_{L-\ell+2,r}}^2}{\sigma_
{h_{\ell,r}}^2}\leq \frac{\sigma_{ h_{i,r}}^2}{\sigma_
{h_{\ell,r}}^2}$. However, the impact of either of these ratios on
the overall average common rate is small compared to that of the
term $\frac{P\sigma^2_{h_{\ell,r}}}{N_0}$ in \eqref{eq:Rcnew},
\eqref{eq:Rcong} and \eqref{eq:Rcnoori}. Thus, the common rate for the proposed scheme will be practically the same as
that of the existing pairing schemes in \cite{Ong:2010} and
\cite{Noori:2012}, which proves Proposition \ref{prop:var_cap}.

\section*{Appendix C\\ Proof of Theorem
\ref{th:sum}}\label{proof:sum}
The achievable rate at the
$(\ell-1)^{th}$ time slot can be obtained from \eqref{eq:R}. Since,
$\frac{\mid h_{i,r}\mid^2}{\mid h_{i,r}\mid^2+\mid
h_{\ell,r}\mid^2}<1$, the achievable rate at the $(\ell-1)^{th}$
time slot will be determined by the achievable rate at the
corresponding time slot in the multiple access phase. Then,
obtaining the achievable rate in all the time slots and adding them
results into \eqref{eq:Rs}. The detailed steps are omitted here for
the sake of brevity.

\section*{Appendix D\\ Proof of Propositions
\ref{prop:eq_sum}$-$\ref{prop:var_sum}}\label{app:sum}

\noindent{\textit{Proof of Proposition 4:}} For the equal average channel gain scenario, $\sigma^2_{
h_{i,r}}=\sigma^2_{h_{\ell,r}}=\sigma^2_{
h_{\ell-1,r}}=\sigma^2_{h_{L-\ell+2,r}}$. Thus, the sum rates
expressed by \eqref{eq:Rsavg}, \eqref{eq:Rsongavg} and
\eqref{eq:Rsnooriavg} become the same for all the three pairing
schemes, which proves Proposition \ref{prop:eq_sum}.

\noindent{\textit{Proof of Proposition 5:}} For the unequal average channel
gain scenario, if the common user is made to transmit at all the
time slots with scaled power, the sum rate can be obtained from
\eqref{eq:Rsavg} with $\sigma^2_{h_{i,r}}$ scaled by $L-1$ and $P$
replaced with $(2L-2){P}$. In this case, the average sum rate in
\eqref{eq:Rsavg} becomes
{\fontsize{11}{13.2}\selectfont
\begin{align}\label{eq:Rsunequal}
E[R_s]&=\frac{1}{2(L-1)}\sum_{\ell=1, \ell\neq
i}^{L}\left(\log\left(\frac{1}{1+\frac{(L-1)\sigma^2_
{h_{\ell,r}}}{\sigma^2_{h_{i,r}}}}+\frac{(2L-2)P\sigma^2_{
h_{i,r}}}{(L-1)N_0}\right)+\log\left(\frac{1}{1+\frac{\sigma^2_
{h_{i,r}}}{(L-1)\sigma^2_{h_{\ell,r}}}}+\frac{(2L-2)P\sigma^2_
{h_{\ell,r}}}{N_0}\right)\right).
\end{align}
\par}
%
\normalsize
Comparing \eqref{eq:Rsunequal} and \eqref{eq:Rsongavg} shows that
$2\sigma^2_{h_{i,r}}>\sigma^2_{h_{\ell-1,r}}$ and
$(2L-2)\sigma^2_{h_{\ell,r}}>\sigma^2_{h_{\ell,r}}$. 
In a similar manner, it can be shown
that the average sum rate of the proposed scheme is larger than that
of the scheme in \cite{Noori:2012}. This completes the proof for
Proposition \ref{prop:uneq_sum}.

\noindent{\textit{Proof of Proposition 6:}} For the variable average
channel gain scenario, we have $\sigma^2_ {h_{i,r}}\geq
\sigma^2_{h_{\ell-1,r}}$. Hence, it is clear that
$\sum_{\ell=1,\ell\neq
i}^{L}\sigma^2_{h_{i,r}}>\sum_{\ell=2}^{L}\sigma^2_ {h_{\ell-1,r}}$.
Similarly, it can be shown that $\sum_{\ell=1,\ell\neq
i}^{L}\sigma^2_{h_{i,r}}>\sum_{\ell=2}^{L}\sigma^2_
{h_{L-\ell+2,r}}$. Thus the proposed pairing scheme will have a
larger average sum rate (given by \eqref{eq:Rsavg}), compared to
that of the pairing schemes in \cite{Ong:2010} and \cite{Noori:2012}
(given by \eqref{eq:Rsongavg} and \eqref{eq:Rsnooriavg},
respectively). This proves Proposition \ref{prop:var_sum}.
\section*{Appendix E\\Derivation of $P_{\sqrt{M}-PAM,NC}(i,m)$ in \eqref{eq:df2waycommon}}
In this appendix, we derive the probability of incorrectly decoding a PAM network coded signal by building on the symbol mapping idea in \cite{Ronald:2013}. We detail the necessary steps to obtain an exact expression for use in the analysis.

We assume $\sqrt{M}$-PAM signals at the $i^{th}$ and the $m^{th}$ users, such that the users' signals can take values from the set $\mathcal{S}=\{\pm1,\pm3,...,\pm(\sqrt{M}-1)\}$ and we denote each element of the set $\mathcal{S}$ as $s$. The true network coded signal resulting from the sum of the $\sqrt{M}$-PAM signals have a constellation with $(2\sqrt{M}-1)$ points, which takes values from the set $\mathcal{S}_{NC}=\{0,\pm2,...,\pm(2\sqrt{M}-2)\}$. 

In a noiseless environment, the relay maps the network coded signal to a $\sqrt{M}$-PAM signal $s$ in such a way that the same network coded signal is not mapped to different elements of $\mathcal{S}$ (i.e., there is no ambiguity). This can be ensured by mapping the network coded signal into modulo-$\sqrt{M}$ sum of the actual symbols at the $i^{th}$ and the $m^{th}$ user. In a noisy environment, the relay maps the network coded signal into $\hat{s}$ and broadcasts to the users, who decode the signal as $\dhat{s}$. The end-to-end probability of incorrectly detecting a network-coded signal resulting from $\sqrt{M}$-PAM signals, can be obtained from the sum of the off-diagonal elements of the product of two $\sqrt{M}\times\sqrt{M}$ matrices $C$ and $D$, with elements $c_{p,q}=P(\hat{s}=q|s=p)$ and $d_{p',q'}=P(\dhat{s}=q'|\hat{s}=p')$, respectively, where $p,q,p',q'\in[0,\sqrt{M}-1]$, multiplied by the factor $\sqrt{M}$. That is,
{\fontsize{11}{13.2}\selectfont
\begin{equation}\label{eq:pam_nc}
P_{\sqrt{M}-PAM,NC}(i,m)=\frac{1}{\sqrt{M}}\left(\sum_{p,q=0}^{\sqrt{M}-1}c_{p,q}\sum_{p',q'=0,p'\neq p, q'\neq q}^{\sqrt{M}-1}d_{p',q'}\right)
\end{equation}
\par}

\noindent The coefficients $c_{p,q}$ can be obtained by calculating the probability that the signal received at the relay whose mean (which takes value from the set $\mathcal{S}_{NC}$) should be mapped to $s=p$, falls in the decision region for the signal whose mean is mapped to $s=q$. Thus, $c_{p,q}$ can be expressed as the sum of $Q$-functions, as follows:
{\fontsize{11}{13.2}\selectfont
\begin{align}\label{eq:cpq}
c_{p,q}=\left\{\begin{array}{ll}\sum\limits_{u=1,u=\textrm{odd}}^{2(2\sqrt{M}-2)-1}a_{p,q,u}Q(u\sqrt{\gamma_{r}(i,m)}) & p\neq q\\
1+\sum\limits_{u=1,u=\textrm{odd}}^{2(2\sqrt{M}-2)-1}a_{p,q,u}Q(u\sqrt{\gamma_{r}(i,m)}) & p=q
\end{array}
\right.
\end{align}
\par}

\noindent where, $\gamma_r(i,m)$ represents the SNR of the $i^{th}$
and the $m^{th}$ users' signal at the relay for $M$-QAM modulation and
can be obtained 
as
{\fontsize{11}{13.2}\selectfont
\begin{equation}\label{eq:snrnpqam}
\gamma_{r}(i,m)=\frac{P\min(\mid {h}_{i,r}\mid^2,\mid
{h}_{m,r}\mid^2)}{E_{av}N_0}.
\end{equation}
\par}

\noindent \noindent where $E_{av}$ is the average energy of symbols for $\sqrt{M}$-PAM modulation (e.g., $E_{av}=5$ for $M=16$).

Similarly, the coefficients $d_{p',q'}$ can be obtained by calculating the probability that the signal received at the $i^{th}$ user with mean $s=p'$, falls in the decision region for the signal with mean $s=q'$. Thus,
{\fontsize{11}{13.2}\selectfont
\begin{align}\label{eq:dp'q'}
d_{p',q'}=\left\{\begin{array}{ll}\sum\limits_{v=1,v=\textrm{odd}}^{2(\sqrt{M}-1)-1}b_{p',q',v}Q(v\sqrt{\gamma_{i}}) & p'\neq q'\\
1+\sum\limits_{v=1,v=\textrm{odd}}^{2(\sqrt{M}-1)-1}b_{p',q',v}Q(v\sqrt{\gamma_{i}}) & p'=q'
\end{array}
\right.
\end{align}
\par}

\noindent where $\gamma_i=\frac{P_r\mid h_{r,i}\mid^2}{E_{av}N_0}$
represents the SNR at the $i^{th}$ user. The coefficients $a_{p,q,u}$ and $b_{p',q',v}$ for $M=16$ (or $\sqrt{M}=4$), have been tabulated in Table \ref{table:ab}.


\begin{table}[t]
{\fontsize{11}{13.2}\selectfont
\caption{Illustration of the coefficients $a_{p,q,u}$ and $b_{p',q',v}$ for $M=16$ corresponding to the probability $P(\hat{V}_{i,m}\neq V_{i,m})$ and $P(\dhat{V}_{i,m}\neq \hat{V}_{i,m})$, respectively.}
\label{table:ab} \centering
\begin{tabular}{|c |c |c |c |c |c |c |c |c |c |c |}
\hline
\multicolumn{2}{|c |}{} & \multicolumn{4}{|c |}{$a_{p,q,u}$} & & \multicolumn{4}{|c |}{$b_{p',q',v}$} \\
\hline
$p,p'$ & {\backslashbox{$u$~}{$q$~~}} & $q=0$ & $q=1$ & $q=2$ & $q=3$ & {\backslashbox{$v$~}{$q'$~~}}& $q'=0$ & $q'=1$ & $q'=2$ & $q'=3$\\
\cline{1-11} \multirow{6}{*}{$p=0$} & $u=1$ & $-7/4$ & $1$ & $0$&$3/4$ &  \multirow{2}{*}{$v=1$} & \multirow{2}{*}{$1/4$} & \multirow{2}{*}{$1/4$} & \multirow{2}{*}{$0$}& \multirow{2}{*}{$0$}\\
\cline{2-6}  & $u=3$ &0 &$-1$ &$7/4$  & $-3/4$ & & & & &\\
\cline{2-11}  & $u=5$ &0 & $3/4$& $-1$ &$1/4$ & \multirow{2}{*}{$v=3$} & \multirow{2}{*}{$0$} & \multirow{2}{*}{$-1/4$}& \multirow{2}{*}{$1/4$} & \multirow{2}{*}{$0$}\\
\cline{2-6}  & $u=7$ &$1$ &$-3/4$ &$0$  &$-1/4$ & & & & & \\
\cline{2-11}  & $u=9$ &$-1/4$ &$1/4$ &$0$  &$0$ & \multirow{2}{*}{$v=5$} & \multirow{2}{*}{$0$} & \multirow{2}{*}{$0$} & \multirow{2}{*}{$-1/4$}  & \multirow{2}{*}{$1/4$}\\
\cline{2-6}  & $u=11$ &0 &$-1/4$ & $1/4$ &$0$ & & & & &\\
\cline{1-11}
\multirow{6}{*}{$p=1$} & $u=1$ & $1$ & $1$ & $0$&$0$&  \multirow{2}{*}{$v=1$} & \multirow{2}{*}{$1/4$} & \multirow{2}{*}{$-1/4$} & \multirow{2}{*}{$1/4$}& \multirow{2}{*}{$0$}\\
\cline{2-6}  & $u=3$ &-1/2 &$0$ &$-1/2$  & $1$ & & & & &\\
\cline{2-11}  & $u=5$ &1/2 & $0$& $1/2$ &$-1$ &  \multirow{2}{*}{$v=3$} & \multirow{2}{*}{$-1/4$} & \multirow{2}{*}{$1/4$} & \multirow{2}{*}{$-1/4$}& \multirow{2}{*}{$1/4$}\\
\cline{2-6}  & $u=7$ &$-1/2$ &$1$ &$-1/2$  &$0$ & & & & &\\
\cline{2-11}  & $u=9$ &$1/2$ &$-1$ &$1/2$  &$0$ &  \multirow{2}{*}{$v=5$} & \multirow{2}{*}{$0$} & \multirow{2}{*}{$1/4$} & \multirow{2}{*}{$0$}& \multirow{2}{*}{$-1/4$}\\
\cline{2-6}  & $u=11$ &0 &$0$ & $0$ &$0$ & & & & &\\
\cline{1-11}
\multirow{6}{*}{$p=2$} & $u=1$ & $1$ & $1$ & $-7/4$&$3/4$ &  \multirow{2}{*}{$v=1$} & \multirow{2}{*}{$0$} & \multirow{2}{*}{$1/4$} & \multirow{2}{*}{$-1/4$}& \multirow{2}{*}{$1/4$}\\
\cline{2-6}  & $u=3$ &7/4 &$-1$ &$0$  & $-3/4$ & & & & &\\
\cline{2-11}  & $u=5$ &-1 & $3/4$& $0$ &$1/4$ &  \multirow{2}{*}{$v=3$} & \multirow{2}{*}{$1/4$} & \multirow{2}{*}{$-1/4$} & \multirow{2}{*}{$1/4$}& \multirow{2}{*}{$1/4$}\\
\cline{2-6}  & $u=7$ &$0$ &$-3/4$ &$1$  &$-1/4$ & & & & &\\
\cline{2-11}  & $u=9$ &$0$ &$1/4$ &$-1/4$  &$0$ &  \multirow{2}{*}{$v=5$} & \multirow{2}{*}{$-1/4$} & \multirow{2}{*}{$0$} & \multirow{2}{*}{$1/4$}& \multirow{2}{*}{$0$}\\
\cline{2-6}  & $u=11$ &$1/4$ &$-1/4$ & $0$ &$0$ & & & & &\\
\cline{1-11}
\multirow{6}{*}{$p=3$} & $u=1$ & $1$ & $0$ & $1$&$-2$ &  \multirow{2}{*}{$v=1$} & \multirow{2}{*}{$0$} & \multirow{2}{*}{$0$} & \multirow{2}{*}{$1/4$}& \multirow{2}{*}{$-1/4$}\\
\cline{2-6}  & $u=3$ &-1 &$2$ &$-1$  & $0$ & & & & &\\
\cline{2-11}  & $u=5$ &1 & $-2$& $1$ &$0$ &  \multirow{2}{*}{$v=3$} & \multirow{2}{*}{$1/4$} & \multirow{2}{*}{$1/4$} & \multirow{2}{*}{$-1/4$}& \multirow{2}{*}{$0$}\\
\cline{2-6}  & $u=7$ &$0$ &$0$ &$0$  &$0$ & & & & &\\
\cline{2-11}  & $u=9$ &$0$ &$0$ &$0$  &$0$ &  \multirow{2}{*}{$v=5$} & \multirow{2}{*}{$0$} & \multirow{2}{*}{$-1/4$} & \multirow{2}{*}{$0$}& \multirow{2}{*}{$0$}\\
\cline{2-6}  & $u=11$ &$0$ &$0$ & $0$ &$0$ & & & & &\\
\cline{1-11}
\end{tabular}\par}
\vspace{-10pt}
\end{table}
\vspace{-20pt}
\section*{Appendix F\\ \vspace{-10pt}Proof of Theorem
\ref{th:aber}}\label{app:aber}

The proof follows the steps outlined in \cite{Shama:2012}, which are applicable to any user
pairing scheme. However, for the proposed pairing scheme, we need
to modify these steps to take into account different error
probabilities at the common user and
the other users. The modified steps can be summarized as follows:
\begin{enumerate}
\item Determine the probabilities that the
$i^{th}$ user and the $\ell^{th}$ user incorrectly decode a network
coded message, respectively.

\item Define the possible error
cases for the $k^{th} (k\in[1,L-1])$ error event at the $i^{th}$ and
the $\ell^{th} $user, where the $k^{th}$ error event means that
exactly $k$ number of users' messages are incorrectly decoded.

\item Express the probabilities of
the aforementioned error cases in terms of the probabilities of
incorrectly decoding a network coded message.

\item Combine the probabilities of
different error cases to determine the probability of the $k^{th}$
error event at the $i^{th}$ and the $\ell^{th}$ user.

\item Obtain the expected
probability of all the error events to determine the exact average
SER expression.

\item Apply the high SNR approximation to obtain approximate but accurate average SER expressions.

\end{enumerate}

Now, we illustrate these steps in detail:\\
\textit{\underline{Step-1}}: The probabilities of incorrectly
decoding a network coded message at the $i^{th}$ and the $\ell^{th}$
user are obtained in \eqref{eq:df2waycommon} and
\eqref{eq:df2wayother}, respectively.

\noindent\textit{\underline{Step-2}}: In the proposed pairing
scheme, $k$ error events can occur in two cases

\begin{itemize}
\item $A_k$: If the decoding user incorrectly extracts exactly $k$
users' messages except the $i^{th}$ user's message. That is, the
decoding user ($j^{th}$ user, where $j\in[1,L]$) incorrectly decodes
$k$ network coded messages $V_{i,m_{1}},
V_{i,m_{2}},...,V_{i,m_{k}}$ and correctly decodes the remaining
$L-1-k$ network coded messages, where $m_1,m_2,...,m_k\in[1,L],
m_1\neq m_2\neq...\neq m_k\neq j$.

\item $B_k$: If the decoding user incorrectly decodes exactly $k$
users' messages including the $i^{th}$ user's message. This happens
when the decoding user ($\ell^{th}$ user, where $\ell\in[1,L],
\ell\neq i$ ) incorrectly decodes $V_{i,\ell}$ and correctly decodes
$k-1$ other network coded messages, $V_{i,m_{1}},
V_{i,m_{2}},...,V_{i,m_{k-1}}$ and incorrectly decodes the remaining
$L-1-k$ messages, where $m_1,m_2,...,m_{k-1}\in[1,L], m_1\neq
m_2\neq...\neq m_{k-1}\neq i,\ell$. 
\end{itemize}

Note that, the error case $A_k$ is applicable both for the common
user and the other users. However, case $B_k$ is applicable only for
users except the common user. 

\noindent\textit{\underline{Step-3}}: The probabilities of the
aforementioned error cases for the $i^{th}$ and the $\ell^{th}$
users are
{\fontsize{11}{13.2}\selectfont
\begin{equation}\label{eq:Aki}
P_{i,A_k}=\sum_{m_a=1,m_a\neq
i}^{L}\prod_{a=1}^{k}P_{FDF}(i,m_a)\prod_{m_b=1, m_b\neq m_a,
i}^{L}\{1-P_{FDF}(i,m_b)\}.
\end{equation}
\begin{equation}\label{eq:Akl}
P_{\ell,A_k}=\sum_{m_a=1,m_a\neq
i,\ell}^{L}\prod_{a=1}^{k}P_{FDF}(\ell,m_a)\prod_{m_b=1,m_b\neq
\ell,m_a}^{L}\{1-P_{FDF}(\ell,m_b)\}.
\end{equation}
\vspace{-10pt}
\begin{equation}\label{eq:Bkl}
P_{\ell,B_k}=\left\{\begin{array}{ll}
P_{FDF}(\ell,i)\sum_{m_a=1,m_a\neq
i,\ell}^{L}\prod_{a=1}^{k-1}\{1-P_{FDF}(\ell,m_a)\}\prod_{m_b=1,m_b\neq
i,\ell,m_a}^{L}P_{FDF}(\ell,m_b) &\mbox{$1<k<L-1$}\\
P_{FDF}(\ell,i)\prod_{m_b=1,m_b\neq
i,\ell}^{L}\{1-P_{FDF}(\ell,m_b)\}
&\mbox{$k=1$}\\
P_{FDF}(\ell,i)\sum_{m_a=1,m_a\neq
i,\ell}^{L}\prod_{a=1}^{L-1}\{1-P_{FDF}(\ell,m_a)\} &\mbox{$k=L-1$}.
\end{array}
\right.
\end{equation}\par}

\noindent\textit{\underline{Step-4}}: The probability of $k$ error
events for the $i^{th}$ and the $\ell^{th}$ user can be expressed as
%
\begin{align}\label{eq:kerrors}
P(i,k)=P_{i,A_k},
P(\ell,k)=P_{\ell,A_k}+P_{\ell,B_k}.
\end{align}

\noindent\textit{\underline{Step-5}}: Since, each user decodes $L-1$
other users' messages in an $L$-user MWRN, there are $L-1$ possible
error events. Thus, averaging over all the possible error events,
the average SER at the $i^{th}$ and the $\ell^{th}$ user can be
obtained as:
%
\begin{align}\label{eq:aber_exact}
P_{i,avg}=\frac{1}{L-1}\sum_{k=1}^{L-1}kP_{i,A_k},
P_{j,avg}=\frac{1}{L-1}\sum_{k=1}^{L-1}k(P_{\ell,A_k}+P_{\ell,B_k})
\end{align}

\noindent\textit{\underline{Step-6}}: At high SNR, the higher order
error terms in \eqref{eq:kerrors} can be neglected. Thus,
$P_{i,A_k}\approx0$ and $P_{\ell,A_k}\approx0$ for $k>1$ (see
\eqref{eq:Aki} and \eqref{eq:Akl}). Similarly,
$P_{\ell,B_k}\approx0$ for $k<L-1$ (see \eqref{eq:Bkl}). Thus, at
high SNR, \eqref{eq:aber_exact} can be approximated as
%
%
\begin{align}\label{eq:aber_k}
P_{i,avg}=\frac{1}{L-1}P_{i,A_1},
P_{\ell,avg}=\frac{1}{L-1}\left(P_{\ell,A_1}+(L-1)P_{\ell,B_{L-1}}\right).
\end{align}

In addition, at high SNR, we can approximate the terms $\{1-P_{FDF}(i,m_b)\}$,
$\{1-P_{FDF}(\ell,m_b)\}$ and $\{1-P_{FDF}(\ell,m_a)\}$ in
\eqref{eq:Aki}, \eqref{eq:Akl} and \eqref{eq:Bkl} to be $1$. Thus,
substituting \eqref{eq:Aki}, \eqref{eq:Akl} and \eqref{eq:Bkl} in
\eqref{eq:aber_k}, the average SER at the $i^{th}$ and the
$\ell^{th}$ user at high SNR can be expressed as
%
{\fontsize{11}{13.2}\selectfont
\begin{align}
P_{i,avg}=\frac{1}{L-1}{\sum_{m_1=1,m_1\neq
i}^{L}P_{FDF}(i,m_1)},
P_{\ell,avg}=\frac{1}{L-1}\left({\sum_{m_1=1,m_1\neq
i,\ell}^{L}P_{FDF}(\ell,m_1)+(L-1)P_{FDF}(\ell,i)}\right).\nonumber
\end{align}\par}

Finally, replacing $m_1$ with $m$ in the above equation completes
the proof.
\vspace{-10pt}
\section*{Appendix G\\ \vspace{-10pt}
Proof of Propositions
\ref{prop:eq_aber}$-$\ref{prop:var_aber}}\label{app:aber_comp}

\noindent{\textit{Proof of Proposition 7:}} For the equal average channel gain scenario,
the error probabilities
$P_{FDF}(j,1)=P_{FDF}(j,2)=...=P_{FDF}(j,L-1)=P_{FDF}$ for all
$j\in[1,L]$. Thus, the average SER expressions in
\eqref{eq:abercommon} and \eqref{eq:aberother} for the proposed
pairing scheme can be simplified as:
%
\begin{align}\label{eq:abereql}
P_{i,avg}=P_{FDF},
P_{\ell,avg}=\left(\frac{2L-3}{L-1}\right)P_{FDF}.
\end{align}

The average SER for the scheme in \cite{Ong:2010} can be given by
\cite{Shama:2012}:
\vspace{-10pt}
\begin{equation}\label{eq:dfavg}
P_{avg}=\frac{L}{2}P_{FDF}.
\end{equation}

Comparing \eqref{eq:abereql} and \eqref{eq:dfavg}, we arrive at Proposition \ref{prop:eq_aber}.

\noindent{\textit{Proof of Proposition 8:}} For the unequal average channel
gain scenario, the average SER expressions for the proposed pairing
scheme is given by \eqref{eq:abercommon} and \eqref{eq:aberother},
with $\gamma_{r}(i,m)=\frac{(2L-2)P\min\left(\frac{\mid
{h}_{i,r}\mid^2}{L-1},\mid {h}_{m,r}\mid^2\right)}{5N_0}$ and
$\gamma_i=\frac{(2L-2)P_r\mid h_{i,r}\mid^2}{5N_0}$. For the scheme
in \cite{Ong:2010}, the average SER at the $j^{th} (j\in[1,L])$ user
can be written as
\begin{equation}\label{eq:aberong}
P_{j,avg}=\frac{1}{L-1}\sum_{m=1}^{L-1}mP_{FDF}(j,m),
\end{equation}

\noindent where
\vspace{-10pt}
{\fontsize{11}{13.2}\selectfont
\begin{equation}\label{eq:df2wayong}
P_{FDF}(j,m)=1-(1-P_{\sqrt{M}-PAM,NC}(j,m))^2,
\end{equation}
\par}

\noindent with $\gamma_{r}(m)=\frac{P\min(\mid {h}_{m,r}\mid^2,\mid
{h}_{m+1,r}\mid^2)}{5N_0}$ and $\gamma_j=\frac{P_r\mid
h_{j,r}\mid^2}{5N_0}$ in \eqref{eq:cpq} and \eqref{eq:dp'q'}, respectively. Now we consider two cases:
\begin{itemize}
\item \textit{case 1}: $E[\frac{\mid h_{i,r}\mid^2}{L-1}]>E[\mid
h_{m,r}\mid^2]$. In this case,
\vspace{-5pt}
{\fontsize{11}{13.2}\selectfont
\begin{align}
&E\left[\min\left(\frac{(2L-2)P\mid
h_{i,r}\mid^2}{5(L-1)N_0},\frac{(2L-2)P\mid
h_{m,r}\mid^2}{5N_0}\right)\right]\nonumber\\&\leq
\min\left(E\left[\frac{(2L-2)P\mid
h_{i,r}\mid^2}{5(L-1)N_0}\right],E\left[\frac{(2L-2)P\mid
h_{m,r}\mid^2}{5N_0}\right]\right)=E\left[\frac{(2L-2)P\mid h_{m,r}\mid^2}{5N_0}\right]\nonumber\\
&\geq \min\left(E\left[\frac{P\mid
h_{m,r}\mid^2}{5N_0}\right],E\left[\frac{P\mid
h_{m+1,r}\mid^2}{5N_0}\right]\right)\geq E\left[\min\left(\frac{P\mid h_{m,r}\mid^2}{5N_0},\frac{P\mid
h_{m+1,r}\mid^2}{5N_0}\right)\right].
\end{align}
\par}

Thus, $E[\gamma_r(i,m)]\geq E[\gamma_r(m)]$.

\item \textit{case 2}: $E[\frac{\mid h_{i,r}\mid^2}{L-1}]<E[\mid
h_{m,r}\mid^2]$. In this case, $E\left[\min\left(\frac{(2L-2)P\mid
h_{i,r}\mid^2}{5(L-1)N_0},\frac{(2L-2)P\mid
h_{m,r}\mid^2}{5N_0}\right)\right]\leq E[\frac{(2L-2)P\mid
h_{i,r}\mid^2}{5(L-1)N_0}]$ and since, $\mid h_{i,r}\mid^2>\mid
h_{m,r}\mid^2, \mid h_{m+1,r}\mid^2$, $E\left[\min\left(\frac{\mid
h_{m,r}\mid^2}{5N_0},\frac{\mid h_{m+1,r}\mid^2}{5N_0}\right)\right]$
$\leq E\left[\frac{(2L-2)P\mid h_{i,r}\mid^2}{5(L-1)N_0}\right]$.
Thus, $E[\gamma_r(i,m)]\geq E[\gamma_r(m)]$.
\end{itemize}

From the above cases, the probability $P_{FDF}(i,m)$ and
$P_{FDF}(\ell,m)$ for the proposed scheme would be larger than
$P_{FDF}(j,m)$ for scheme \cite{Ong:2010}. Thus, comparing
\eqref{eq:abercommon}, \eqref{eq:aberother} and \eqref{eq:aberong}
shows that the average SER for the proposed scheme would be smaller
than that for scheme \cite{Ong:2010}. This proves Proposition
\ref{prop:uneq_aber}.

\noindent{\textit{Proof of Proposition 9:}} For the variable average channel gain scenario, the average SER expression for the proposed pairing scheme is
given by \eqref{eq:abercommon} and \eqref{eq:aberother}. The average
SER for the pairing scheme in \cite{Ong:2010} is the same as in
\eqref{eq:aberong}. Now, comparing $P_{FDF}(i,m)$ (from
\eqref{eq:df2waycommon}), $P_{FDF}(\ell,m)$ (from
\eqref{eq:df2wayother}) and $P_{FDF}(j,m)$ (from
\eqref{eq:df2wayong}) shows that the only terms which are different
in all these probabilities are $\gamma_r(i,m)$ and $\gamma_r(m)$.
Note that, if $E[\mid {h}_{i,r}\mid^2]>E[\mid {h}_{m+1,r}\mid^2]$,
then $E[\min(\mid {h}_{i,r}\mid^2,\mid {h}_{m,r}\mid^2)]\geq
E[\min(\mid {h}_{m+1,r}\mid^2,\mid {h}_{m,r}\mid^2)]$. Thus,
$E[\gamma_r(i,m)]\geq E[\gamma_r(m)]$ and in effect, from
\eqref{eq:df2waycommon}, \eqref{eq:df2wayother} and
\eqref{eq:df2wayong}, the error probability for the new pairing
scheme would be less than that for scheme \cite{Ong:2010}. As a
result, the average SER for the proposed scheme 
is less than that of
scheme \cite{Ong:2010} (in \eqref{eq:aberong}) for both $j=i$ and
$j=\ell$, which proves Proposition \ref{prop:var_aber}.
\vspace{-10pt}
\bibliographystyle{IEEEtran}





\end{document}